\documentclass[a4paper,UKenglish,cleveref]{lipics-v2019}
\pdfoutput=1

\usepackage{mathtools} 
\usepackage{stmaryrd} 
\usepackage{tikz}
\usetikzlibrary{matrix,arrows,calc,backgrounds,shapes,shapes.geometric,patterns,positioning,fit,decorations.pathmorphing,decorations.pathreplacing,shadows}

\usepackage{enumitem}

\usepackage[noend,noline,ruled]{algorithm2e}
\SetKwInput{KwInput}{Input}

\usepackage[calc]{adjustbox}
\usepackage{wrapfig}

\usepackage{todonotes}

\newcommand{\cutout}[1]{}

\theoremstyle{definition}
\newtheorem{Definition}[theorem]{Definition}
\newtheorem*{open}{Open Problem}

\newcommand{\Iff}{\Longleftrightarrow}
\newcommand{\ra}{\rightarrow}
\newcommand{\impl}{\rightarrow}

\newcommand{\E}{\exists}
\newcommand{\A}{\forall}

\newcommand{\QD}{\mathsf{D}} 
\newcommand{\QE}{\mathsf{E}} 
\DeclareMathOperator{\QGA}{\rotatebox[origin=c]{180}{$\forall$}} 
\DeclareMathOperator{\QGE}{\rotatebox[origin=c]{180}{$\exists$}} 
\newcommand{\comvar}[2]{#1 \operatorname{\overline{\cap}} #2}
\DeclareMathOperator{\tr}{\textsf{tr}}

\renewcommand{\phi}{\varphi}
\renewcommand{\theta}{\vartheta}
\renewcommand{\emptyset}{\varnothing}
\renewcommand{\epsilon}{\varepsilon}

\newcommand{\FO}{\normalfont{\textsf{FO}}}
\newcommand{\GF}{\normalfont{\textsf{GF}}}
\newcommand{\CGF}{\normalfont{\textsf{CGF}}}
\newcommand{\ML}{\normalfont{\textsf{ML}}}
\newcommand{\LFD}{\normalfont{\textsf{LFD}}}
\renewcommand{\L}{\normalfont{\textsf{L}}}
\newcommand{\On}{\mathrm{On}}

\newcommand*{\ext}[1]{[\![ #1 ]\!]}

\DeclareMathOperator{\ar}{ar} 
\DeclareMathOperator{\qr}{qr} 
\DeclareMathOperator{\Free}{free} 
\DeclareMathOperator{\Th}{Th} 

\newcommand{\dep}{\operatorname{dep}}
\newcommand{\Ind}{\operatorname{Ind}}

\newcommand{\excl}{\mathbin{\operatorname{|}}}
\newcommand{\indep}{\perp}
\newcommand{\incl}{\subseteq}
\newcommand{\anon}{\Upsilon}

\newcommand{\PSPACE}{\textsc{Pspace}}
\newcommand{\PTIME}{\textsc{Ptime}}

\newcommand{\ALOGSPACE}{\textsc{Alogspace}}
\newcommand{\APTIME}{\textsc{Aptime}}

\newcommand{\accept}{\textsc{Accept}}
\newcommand{\reject}{\textsc{Reject}}

\DeclareMathOperator{\MC}{MC}
\DeclareMathOperator{\MCL}{MC_{list}}
\DeclareMathOperator{\MCF}{MC_{full}}
\DeclareMathOperator{\MCFK}{MC^{\textit{k}}_{full}}
\DeclareMathOperator{\MCFO}{MC_{formula}}
\DeclareMathOperator{\MCFOK}{MC^{\textit{k}}_{formula}}
\DeclareMathOperator{\MCFOBK}{MC^{\textit{B,k}}_{formula}}

\newcommand*{\tup}[1]{\bar{#1}}
\newcommand{\ta}{\tup a}
\newcommand{\tb}{\tup b}
\newcommand{\tc}{\tup c}
\newcommand{\tx}{\tup x}
\newcommand{\ty}{\tup y}
\newcommand{\tz}{\tup z}
\newcommand{\tv}{\tup v}
\newcommand{\und}[1]{\underline{#1}}
\newcommand{\restr}{\upharpoonright}
\newcommand{\defiff}{\mathbin{:\,\Longleftrightarrow}}

\newcommand{\bbN}{{\mathbb N}} 
\newcommand{\cO}{\mathcal{O}} 
\newcommand{\Vv}{{\cal V}} 

\newcommand{\frA}{{\mathfrak A}}
\newcommand{\frB}{{\mathfrak B}}
\newcommand{\frC}{{\mathfrak C}}
\newcommand{\frM}{\mathfrak{M}} 
\newcommand{\frN}{\mathfrak{N}} 
\newcommand{\TM}{T}
\newcommand{\TN}{T'}
\newcommand{\M}{(\frM,\TM)}
\newcommand{\N}{(\frN,\TN)}
\newcommand{\TMS}{{s}}
\newcommand{\TMT}{{t}}
\newcommand{\TMU}{{u}}
\newcommand{\TNS}{{s'}}
\newcommand{\TNT}{{t'}}

\newcommand{\TTMS}{{\tup \TMS}}
\newcommand{\TTMT}{{\tup \TMT}}

\newcommand{\TTNS}{{\tup \TNS}}
\newcommand{\TTNT}{{\tup \TNT}}

\title{Logics of Dependence and Independence: The Local Variants}

\author{Erich Grädel}{RWTH Aachen University, Germany}{graedel@logic.rwth-aachen.de}{}{}
\author{Phil Pützstück}{RWTH Aachen University, Germany}{phil.puetzstueck@rwth-aachen.de}{}{}

\authorrunning{E.~Grädel and P.~Pützstück}

\Copyright{E.~Grädel and P.~Pützstück}


\nolinenumbers 
\hideLIPIcs  

\begin{document}

\maketitle

\begin{abstract}
Modern logics of dependence and independence are based on team semantics, which means that formulae
are evaluated not on a single assignment of values to variables, but on a set of such assignments, called a team.
This leads to high expressive power, on the level of existential second-order logic.
As an alternative, Baltag and van Benthem have proposed a local variant of dependence logic, 
called logic of functional dependence  ($\LFD$).
While its semantics is also based on a team, the formulae are evaluated \emph{locally} on just one of its 
assignments, and the team just serves as the supply of the possible assignments
that are taken into account in the evaluation process. This logic thus 
relies on the modal perspective of generalized assignments semantics, and can be seen as a fragment of first-order
logic.  For the variant of $\LFD$ without equality, the satisfiability problem is decidable.

We extend the idea of localising logics of dependence and independence in a systematic way, taking 
into account local variants of standard atomic dependency properties: besides dependence and independence,
also inclusion, exclusion, and anonymity. We study model-theoretic and algorithmic questions
of the localised logics, and also resolve some of the questions that had been left open by Baltag and van Benthem.
In particular, we study decidability issues of the local logics, and prove that 
satisfiability of $\LFD$ with equality is undecidable.
Further, we establish characterisation theorems via appropriate notions of bisimulation and study the
complexity of model checking problems for these logics.
\end{abstract}

\section{Introduction}

Following work that has been initiated by Hodges \cite{Hodges97} and 
Väänänen \cite{Vaananen07}, modern logics of dependence and independence are 
generally based on \emph{team semantics}. This means that a formula $\phi(x_1,\dots,x_m)$
in such a logic is not evaluated on a single assignment $s:\{x_1,\dots,x_m\}\ra A$ of values to 
the free variables,
but for a \emph{set} of such assignments, which is called a \emph{team}. In these logics,
dependence and independence statements about variables, such as \emph{``$y$ depends on $x$''} or
\emph{``$x$ and $y$ are independent''} are considered as \emph{atomic properties of teams}.
Besides dependence and independence atoms, further atomic team properties have been
considered: inclusion, exclusion, anonymity, conditional independence, and others.
A crucial feature of the logics with team semantics is that they manipulate second-order objects
by first-order syntax. The evaluation of a formula on a team can be considered as a dynamic process
(or game) that modifies the team while moving from the formula
through its syntax tree to the atoms or, equivalently, as an annotation of the syntax tree of the formula by teams.
The expressive power of logics with team semantics is typically on the level of powerful fragments
of existential second-order logic. More precisely, with a logic $L$ with team semantics, 
which syntactically extends first-order logic by certain atomic team properties, one can associate
a fragment $F\subseteq \Sigma^1_1$ of existential-second-order, such that every formula
$\phi(\bar x)\in L$ (of vocabulary $\tau$) is equivalent to a \emph{sentence}
$\psi\in F$ (of vocabulary $\tau\cup\{\TM\}$, where $\TM$ is a predicate for the team),
such that $\phi(\bar x)$ is true in a structure $\frM$ with a team $\TM$, if, and only if
the expanded structure $\M$ is a model of $\psi$. Understanding the expressive power of a logic
with team semantics thus means to identify the fragment $F\subseteq \Sigma^1_1$ such that
these equivalences hold in both directions. Here are some of the most important results of this kind: 
\begin{itemize}
    \item Dependence logic $\FO(\dep)$ and exclusion logic $\FO({}\excl{})$
        are equivalent to the fragment of  $\Sigma^1_1$-sentences in which the
        predicate for the team appears only negatively \cite{KontinenVaa09}.

    \item Inclusion logic $\FO(\subseteq)$ and anonymity logic $\FO(\anon)$ are
        equivalent to sentences of form $\forall \bar x (\TM\bar x\rightarrow \psi(\TM,\bar x))$
        in the posGFP-fragment of least fixed-point logic,
        such that $\TM$ occurs only positively in $\psi$ \cite{GallianiHel13}.

    \item Independence logic $\FO(\indep)$ and inclusion-exclusion logic
        $\FO(\incl,{}\excl{})$ are equivalent with full $\Sigma^1_1$
        (and thus can describe all NP-properties of teams) \cite{Galliani12}.
\end{itemize}
In particular, all of these logics are much more expressive than classical
first-order logic and are of course undecidable for satisfiability.

\medskip 
Recently, Baltag and van Benthem \cite{BaltagBen20} have proposed a different
kind of dependence logic, called logic of functional dependence $\LFD$,
which can be considered as a fragment of first-order logic.
Its semantics is also based on an underlying team of assignments,
but the formulae are evaluated \emph{locally} on just one of these assignments,
while the team serves as the supply of (or one might also say, the restriction for)
the possible assignments that are admitted in the evaluation process of the formula.
This is the modal perspective of generalized assignments semantics,
where not all possible assignments of values to variables are available but
only a given set of them; these may be considered as the possible worlds
in a Kripke style semantics, where neighbouring worlds are assignments that
agree on some subset of the variables. Thus, in an evaluation process of such
a formula, the team remains unchanged, but one moves around between different
assignments in that team. Dependence comes in by means of atoms $D_X y$
(where $y$ is a variable and $X$ is a set of variables), which are true at a
given assignment $s$ if all assignments $t$ in the underlying team that agree
with $s$ on all variables in the set $X$, also agree with $s$ on $y$.
The atom $D_X y$ can thus be read as a \emph{local} dependence of $y$ on $X$
(on the assignments related to the current one by agreement on $X$),
whereas the standard dependence atom $\dep(X,y)$ in Väänänen's dependence logic
\cite{Vaananen07} says that $y$ \emph{globally} depends on $X$ in the team,
in the sense that there is function which, for \emph{all} assignments in the team,
maps the values for $X$ to the value for $y$. There is a further dependence
operator in $\LFD$, of form $\QD_X \phi$, with the meaning that $\phi$ is true at
all assignments that agree with the current one on $X$.
 
Baltag and van Benthem provide a detailed study of many aspects of $\LFD$.
In particular they prove that $\LFD$ can be embedded into classical first-order logic
(with the team of possible assignments as an additional predicate) and they view
$\LFD$ as a \emph{minimal} logic of functional dependence. They prove that the variant
of $\LFD$ without equality is decidable for satisfiability, but leave open the problem
whether this is also the case for $\LFD^=$, the version with equality.
Towards the end of their paper, Baltag and van Benthem also study a local version
of \emph{independence}, by atoms $\Ind_X(Y)$, saying that the values of the variables
in $X$ at the current assignment $s$ do not constrain the values of the variables
in $Y$: for any assignment $t$ in the team there is one with the same $X$-values
as $s$ and the same $Y$-values as $t$. Interestingly, this local notion of independence
is not symmetric: $\Ind_X(Y)$ does not imply $\Ind_Y(X)$, contrary to the global
independence atoms $X\bot Y$ in the independence logic of Grädel and Väänänen \cite{GraedelVaa13}.
Baltag and van Benthem also observe that this local independence logic is undecidable,
even without equality.

In this paper, we resolve some of the problems left open by Baltag and van Benthem.
In particular, we prove that $\LFD^=$ is undecidable, and we establish a
characterisation theorem via an appropriate notion of bisimulation. 

Further, we study the idea of \emph{localising} logics of dependence and
independence in a more systematic way. We consider local variants of all the
standard atomic dependency notions for teams; beyond dependence and independence,
this includes anonymity, inclusion, and exclusion. In this setting, it turns out
that anonymity is just the negation of dependence, and inclusion is the negation
of exclusion (which is not the case for the global team semantical variants of these atoms).
This also suggests to look more closely at the role of negation.
If, as in global logics of dependence and independence, we insist that formulae
are written in negation normal form, and that negation is applied to classical
atoms only, but not to atomic dependencies, then by considering local dependence
atoms $D_X y$ and local anonymity atoms $\anon_Xy\equiv\neg D_X y$ separately
we obtain two even ``more minimal'' logics than $\LFD$, whose common extension coincides with $\LFD$.

\section{Localised variants of logics of dependence and independence}

\subsection{Global and local atomic properties of teams}

A team $\TM$ is a set of assignments $\TMS:\Vv\ra M$ with a common domain $\Vv$ of variables
and a common set of values $M$, typically the universe of a structure. 
For a tuple $\tx$ of variables from $\Vv$ we write $\TM(\tx) \coloneqq \{\TMS(\tx) \mid \TMS \in \TM\}$
for the set of values of $\tx$ in the team $\TM$.
Whenever $\Vv$ is finite, it is convenient to fix
some enumeration $\tv$ of $\Vv$ so we can denote assignments $\TMS$
by their tuple of values $\TTMS = \TMS(\tv)$.
Conversely, given some fitting tuple of values $\ta$,
the notation $\tv \mapsto \ta$ represents an assignment $\TMS$ with $\TMS(\tv) = \ta$.
This allows us to view $\TM$ as a relation in the usual sense.

\medskip 
For the purpose of this paper, we call the standard team semantical atoms such as
dependence, independence, anonymity, inclusion and exclusion \emph{global atoms}.
A global atom $\alpha$, over a set of variables $\Vv(\alpha)$ defines, 
for every set $M$ of values, the extension $\ext{\alpha}^M$ of all teams $\TM$ such that 
$M \models_{\TM} \alpha$. 
We require that the domain of $\TM$ contains at least $\Vv(\alpha)$, and that the
truth of $\alpha$ in $\TM$ only depends on the variables in $\Vv(\alpha)$, i.e., 
$M\models_{\TM} \alpha$ if, and only if, $M \models_{\TM\restriction\Vv(\alpha)} \alpha$.
The most important global atoms are:
\def\tmpWidth{\widthof{$M \models_{\TM} \dep(\tx,y)$}}
\begin{description}[leftmargin=!,labelwidth=\widthof{Independence: }]
    \item[Dependence:]
        \makebox[\tmpWidth]{$M \models_{\TM} \dep(\tx, y)$\hfill}
        $\defiff\ (\A \TMS\in \TM)(\A \TMT\in \TM)\ \TMS(\tx) = \TMT(\tx) \ra \TMS(y) = \TMT(y)$;
    \item[Inclusion:]
        \makebox[\tmpWidth]{$M \models_{\TM} \tx\incl\ty$\hfill}
        $\defiff\ (\forall \TMS\in \TM)(\exists \TMT\in \TM)\ \TMS(\tx) = \TMT(\ty)$;
            \item[Exclusion:]
        \makebox[\tmpWidth]{$M \models_\TM \tx \excl \ty$\hfill}
        $\defiff\ (\A \TMS\in \TM)(\A \TMT \in \TM)\ \TMS(\tx) \neq \TMT(\ty)$;
    \item[Anonymity:]
        \makebox[\tmpWidth]{$M \models_\TM \tx\anon y$\hfill}
        $\defiff\ (\A \TMS\in \TM)(\E \TMT\in \TM)\ \TMS(\tx)=\TMT(\tx) \land \TMS(y)\neq \TMT(y)$;
     \item[Independence:]
        \makebox[\tmpWidth]{$M \models_\TM  \tx\indep \ty$\hfill}
        $\defiff\ (\A \TMS,\TMT\in \TM)(\E \TMU\in \TM)\ \TMS(\tx) = \TMU(\tx) \land \TMT(\ty) = \TMU(\ty)$.
\end{description}
\def\tmpWidth{\relax}

Notice that all these global atoms are defined by a universal expression of form
$(\A \TMS\in \TM)\beta$ where $\beta$ is a statement about equalities and inequalities between 
values of $\TMS$ and values of other $\TMT \in \TM$. 

A \emph{local} team semantical atom $\beta$ instead is evaluated on a local assignment $\TMS$ 
of some fixed underlying team $\TM$. Thus,  the extension $\ext{\beta}^{M,\TM}$ of $\beta$
on a set $M$ of values and a team $\TM$ of assignments $\TMT: \Vv\ra M$, is
the set of all $\TMS\in \TM$ such that $(M,\TM)\models_\TMS \beta$.
We say that $\beta$ is the local variant of the global atom $\alpha$ if, for any 
set of values $M$ and any team $\TM$, we have that
\[
    M \models_\TM \alpha \ \Iff\ (M,\TM)\models_\TMS \beta \text{ for all } \TMS \in \TM.
\]

The following example gives a local variant of dependence.

\begin{example}\label{ex:local-dep}
    Consider the team over $\Vv = \{x,y,z\}$ represented by the following table:
    \begin{center}
    \begin{tabular}{ccc}
        $x$ & $y$ & $z$\\\hline
         0  &  0  & 0\\
         1  &  1  & 0\\
         1  &  2  & 1\\
         2  &  2  & 1
    \end{tabular}.
    \end{center}
    We say that \emph{$x$ locally depends on $y$ at the assignment
    $(1,1,0)$}, because fixing $y$ to be its value
    \emph{in this specific assignment},
    also fixes $x$ to its value in said assignment.
    In other words, fixing $y$ to $1$ causes $x$ to be fixed to 1.
    As another example, $y$ locally depends on $z$
    at the assignments $(1,2,1)$ and $(2,2,1)$,
    because $z=1$ entails $y=2$ in the above table.
    However, $y$ does not depend on $x$ at $(1,1,0)$.

    Notice that the global functional dependence
    is the universal closure of this local dependence;
    $y$ depends on $x$ in the whole team if, and only if,
    $y$ locally depends on $x$ at every assignment in the team.
\end{example}

The local variants of the standard global dependency atoms are the following:
\def\tmpWidth{\widthof{$(M,\TM) \models_s \Ind_{\tx}(\ty)$}}
\begin{description}[leftmargin=!,labelwidth=\widthof{Independence: }]
    \item[Dependence:]
        \makebox[\tmpWidth]{$(M,\TM)\models_\TMS D_{\tx} y$\hfill}
        $\defiff\ (\A \TMT\in \TM)\ \TMS(\tx) = \TMT(\tx) \ra \TMS(y) = \TMT(y)$; 

    \item[Inclusion:]
        \makebox[\tmpWidth]{$(M,\TM)\models_\TMS \tx \in \ty$\hfill}
        $\defiff\ \TMS(\tx)\in \TM(\ty)$;

    \item[Exclusion:]
        \makebox[\tmpWidth]{$(M,\TM)\models_\TMS \tx \notin \ty$\hfill}
        $\defiff\ \TMS(\tx) \notin \TM(\ty)$;

    \item[Anonymity:]
        \makebox[\tmpWidth]{$(M,\TM)\models_\TMS \anon_{\tx} y$\hfill}
        $\defiff\ (\E \TMT\in \TM)\ \TMS(\tx) = \TMT(\tx) \land \TMS(y) \neq \TMT(y)$;

    \item[Independence:]
        \makebox[\tmpWidth]{$(M,\TM)\models_\TMS \Ind_{\tx}(\ty)$\hfill}
        $\defiff\ (\A \TMT \in \TM)(\E \TMU\in \TM)\ \TMS(\tx) = \TMU(\tx) \land \TMT(\ty) = \TMU(\ty)$.
\end{description}
\def\tmpWidth{\relax}

Notice that there are some striking differences between the properties of local and global
atoms. For the local atoms, anonymity and dependence, as well as inclusion and exclusion,
are directly related via negation, which is not the case for the global atoms.
Further, local independence is not symmetric in $\tx$ and $\ty$.

\subsection{The local dependence logic LFD}

We now describe the \emph{logic of functional dependence} $\LFD$, as introduced
by Baltag and van Benthem \cite{BaltagBen20}, and recall the main results
they proved on this logic.
We shall then propose a slightly different presentation of local
logics on teams, of which $\LFD$ is one special case.

For a tuple $\ta = (a_1,\dots,a_n) \in A^n$,
we denote by $[\ta]\coloneqq \{a_1,\dots,a_n\}$ the set of its components.
Given some function $f$ with $[\ta] \subseteq \operatorname{dom}(f)$,
we also write $f(\ta) \coloneqq (f(a_1),\dots,f(a_n))$.

\begin{Definition}
A type $(\tau,\Vv)$ consists of a relational vocabulary
$\tau$ and a set $\Vv$ of variables.
If both $\tau$ and $\Vv$ are finite, we say $(\tau,\Vv)$ is a finite type.
A dependence model $\M$ of type $(\tau,\Vv)$
consists of a $\tau$-structure $\frM$ with universe $M$,
and a nonempty team $\TM$ with domain $\Vv$ and co-domain $M$. 
Pointed dependence models $\M,\TMS$ further distinguish a ``current'' assignment $\TMS \in \TM$.
\end{Definition}

\begin{Definition}\label{def:lfd-syntax}
The syntax of formulae in $\LFD(\tau,\Vv)$ is given by
    \[
        \phi\quad \Coloneqq
            \quad R\tx
            \quad\mid\quad D_Xy
            \quad\mid\quad \lnot \phi
            \quad\mid\quad \phi \land \phi
            \quad\mid\quad \QD_X\phi,
    \]
    where $R\in \tau$ is a relation symbol,
    $\tx $ is a tuple of variables of appropriate length,
    $X \subseteq \Vv$ is \emph{finite}, and $y \in \Vv$.
    Further, $\LFD^=$ is the extension of $\LFD$ by equality atoms $x=y$. We shall also make use of the
    following notations:
     \begin{itemize}
        \item The dual of $\QD_X\phi$ is defined as $\QE_X\phi \coloneqq \lnot \QD_X\lnot \phi$.
        \item For the special case $X = \emptyset$ we use $\QGA\phi \coloneqq \QD_\emptyset\phi$ and $\QGE \phi \coloneqq \QE_\emptyset \phi$.
        \item We use $\QD_{\tx}\phi \coloneqq \QD_{[\tx]}\phi$
            and likewise for the other quantifiers.
    \end{itemize}
\end{Definition}

We refer to the $\QD_X$ and $\QE_X$ as dependence quantifiers or modalities of $\LFD$,
and also call $\QGA$ and $\QGE$ global modalities.
As we will often deal with dependence on \emph{sets} of variables,
we will use the notation $\TMS =_X \TMT$ whenever the assignments $\TMS,\TMT \in \TM$ agree
on the set of variables $X \subseteq \Vv$, i.e.~$\TMS(x) = \TMT(x)$ for all $x \in X$.
Note that ${=_\emptyset} = \TM \times \TM$.

\begin{Definition}\label{def:lfd-semantics}
The semantics of $\LFD$ and $\LFD^=$ on a dependence model $\M$ is defined by the usual rules 
for atomic formulae and connectives together with
\begin{alignat*}{3}
    &\M \models_\TMS D_Xy
        &&\qquad\defiff\qquad \TMT =_X \TMS \text{ implies } \TMT =_y \TMS \text{ for all $\TMT \in \TM$},\\
    &\M\models_\TMS \QD_X\phi
        &&\qquad \defiff\qquad \M \models_\TMT \phi \text{ for all $\TMT \in \TM$ with $\TMT =_X \TMS$}.
\end{alignat*}
Obviously, $\QE_X,\QGE,\QGA$ then have the expected semantics
\begin{alignat*}{3}
    &\M \models_\TMS \QE_X\phi
        &&\qquad\Iff\qquad \M \models_\TMT \phi \text{ for some $\TMT\in \TM$ with $\TMT =_X \TMS$},\\
    &\M \models_\TMS \QGE\phi
        &&\qquad\Iff\qquad \M \models_\TMT \phi \text{ for some $\TMT\in \TM$},\\
    &\M \models_\TMS \QGA\phi
        &&\qquad\Iff\qquad \M \models_\TMT \phi \text{ for all $\TMT\in \TM$}.
    \end{alignat*}
\end{Definition}

Of course, $\QD_X$ and $\QE_X$ are variants of the traditional quantifiers $\forall$ and
$\exists$, and we briefly want to justify the use of these new quantifiers instead of the
classical ones.
The main point is that the reasoning about free and bound variables becomes more transparent.
Indeed, it is a desirable feature that the meaning of a formula should only depend on its free
variables, in the sense that if $\TMS =_{\Free(\phi)} \TMT$ then it should be the case that
$\M \models_\TMS \phi$ if, and only if, $\M \models_\TMT \phi$.
By defining  $\Free(D_Xy) = \Free(\QD_X\phi) = X$, this is easily seen to be the case.
Using a naive interpretation of traditional quantifiers, saying that $\A y\phi$ holds 
for the assignment $\TMS$ if, and only if $\phi$ holds for all assignments $\TMT$
that agree with $\TMS$ on all variables except $y$, we would get the unwanted    
behaviour\footnote{Often called non-locality, although we do not want to get these two notions of locality mixed up.}
that $\A y \phi$ may depend on variables not occurring free in it. 
Indeed, it is easy to construct a team in which the formula $\forall y Py$
is true at some assignments and false at others, although all assignments obviously
coincide on the free variables of $\A yPy$ (of which there are none).
On the other side, Baltag and van Benthem showed that such
problems do not occur if one translates dependency quantifiers into universal ones 
by $\QD_{\tx}\phi\ \Iff\ \A\tz\phi$,
where $\tz$ is an enumeration of $\Vv \setminus [\tx]$.

\medskip
The semantics of $\LFD$ formalizes a \emph{local} notion
of dependence, as in \cref{ex:local-dep}.
It is also this locality, together with the semantics of the dependence
quantifiers $\QD_X$ and $\QE_X$, which emphasizes the modal character of $\LFD$.
Indeed, notice the similarities to the modalities $\square$ and $\lozenge$ of
propositional modal logic $\ML$; the binary relations $=_X$ on  teams can be viewed as
the accessibility relations of the modality $\QD_X$ and its dual $\QE_X$ on the team.
In this sense, $\M$ can be viewed as a Kripke structure
that has $\TM$ as its universe and accessibility relations $(=_X)$ for $X \subseteq \Vv$.
Note that global functional dependence
is expressible in $\LFD$ via $\QGA D_Xy$,
which guarantees that $X$ locally determines $y$ at every assignment in the team,
i.e.~that $X$ determines $y$ in the whole team.

We shall see that there are important differences between $\LFD$ and $\LFD^=$.
In particular, without equality, we can assume, without loss of generality, that
the teams $\TM$ that we consider are variable-distinguished, i.e. $\TM(x)\cap \TM(y)=\emptyset$
for distinct variables $x,y$.

\begin{proposition}\label{fact:all-dist}
    Let $\M$ be a dependence model with universe $M$ and variables $\Vv$.
    Set $N = M \times \Vv$, $\TN = \{\TNS \mid \TMS \in \TM\}$ where $\TNS(x) = (\TMS(x),x)$,
    and construct $\frN$ with universe $N$ by adapting the relations
    so that $\frN \models_{\TNS} R\tx\ \Iff\ \frM \models_\TMS R\tx$.
    Then $\N$ is variable-distinguished and $\LFD$-equivalent to $\M$.
\end{proposition}

\subsection{A general definition for localised logics}

To consider logics based on other local atoms than the local dependency atom,
we generalise the logic $\LFD$ described in the last subsection.
We will use a basic sublogic $\L$ of relational atoms and boolean connectives.
This is extended by our local atoms and the dependence quantifiers $\QD_X$ and $\QE_X$
introduced in the last section.
The main difference is that we will allow negation only on relational atoms,
to further differentiate between different local atoms.

Now let $\Omega\subseteq\{D, \anon, =, \neq, \in ,\not\in, \Ind,\dots\}$
be a collection of atomic operators, which, applied to appropriate tuples
of variables from $\Vv$, define local atoms $\beta(\tx)$ such as
$x=y$, $x\neq y$, $D_{\tx} y$, $\anon_{\tx} y$, $\tx\in \ty$, $\tx\not\in \ty$, $\Ind_{\tx} \ty$, etc.
We denote the collection of such atoms by $\Omega[\Vv]$.

\begin{Definition}
    The syntax of formulae in the local team logic $\L[\Omega]$ of type $(\tau,\Vv$) is given by
    \[
        \phi\quad \Coloneqq
            R\tx
            \ \mid\ \neg R\tx
            \ \mid\ \beta(\tx)
            \ \mid\ \phi \lor \phi
            \ \mid\ \phi \land \phi
            \ \mid\ \QE_X\phi
            \ \mid\ \QD_X\phi
    \]
    where $R\in \tau$ is a relation symbol,
    $X \subseteq \Vv$ is finite, $\tx $ is a tuple of variables in $\Vv$ of appropriate length,
    and $\beta(\tx)\in\Omega[\Vv]$.
    Formulae are evaluated on dependence models $\M$ consisting of $\tau$-structure $\frM$ and a
    team $\TM$ with domain $\Vv$ and values in $\frM$, locally at some assignment $\TMS \in \TM$.
    The rules to determine whether $\M\models_\TMS \phi$, for $\phi\in L[\Omega]$ extend
    the truth definitions for $\tau$-literals $R\tx$ and $\neg R\tx$,
    and for local team atoms $\beta(\tx)\in\Omega[\Vv]$
    by the standard rules for boolean connectives and
    the rules for dependence quantifiers $\QD_X,\QE_X$ as in \cref{def:lfd-semantics}.
\end{Definition}

The free variables of $=,\neq$ are defined as usual.
For the other atoms,
we set $\Free(D_Xy) = X$,
$\Free(\tx \in \ty) = [\tx]$,
and $\Free(\Ind_{\tx}(\ty)) = [\tx]$,
where the negations of these atoms have the same free variables.

We shall also use the operators  $\QGA$ and $\QGE$ explained above.
Note that we can now state our definition that some atom $\beta$ is the local
variant of a global atom $\alpha$ via
\[
    M \models_\TM \alpha
    \quad\Iff\quad (M,\TM) \models \QGA \beta
    \quad\defiff\quad (M,\TM) \models_\TMS \beta\text{ for all } \TMS \in \TM.
\]
We will thus allow ourselves to use $\alpha$ in $\L[\beta]$,
with the intended semantics $\alpha \equiv \QGA \beta$, e.g.~we can write $\dep(x,y)$ instead
of $\QGA D_xy$ within $\L[D]$.

Obviously, the logics $\LFD$ and $\LFD^=$ of Baltag and van Benthem 
are equivalent to $\L[D,\anon]$ and $\L[D,\anon,=,\neq]$, respectively.

\subsection{The standard translation into first-order logic}\label{subsec:trans}

A dependence model $\M$ of type $(\tau,\Vv)$ with finite $\Vv$,
enumerated as $\Vv=\{v_1\dots,v_n\}$, can be viewed as a structure of vocabulary $\tau\cup\{\TM\}$,
expanding $\frM$ by a predicate (of arity $|\Vv|=n$) for the team.
It is not difficult to see that, for any of the logics $\L[\Omega]$ described above,
there is straightforward translation that associates with every formula
$\phi\in\L[\Omega]$ of type $(\tau,\Vv)$ an equivalent formula $\phi^* \in \FO$
of vocabulary $\tau\cup\{\TM\}$, which means that,
for all $\M$ and all $\TMS\in \TM$, we have that
\[
    \M\models_\TMS \phi \ \Iff\  \M\models \phi^*(\TMS(\tv))
\]
Notice that the semantics on the left side is the semantics of
$\L[\Omega]$ whereas on the right side we have the
standard Tarski semantics of first-order logic. 

For the case of $\LFD$ and $\LFD^=$ this has been called
the \emph{standard translation} in \cite{BaltagBen20},
and it is a straightforward generalisation of the translations of modal logics into first-order logic.  
The rules of the translation are trivial for $\tau$-literals and Boolean connectives.
Quantifiers are translated via relativisation to the team predicate.
Given a formula $\phi$, $X \subseteq \Vv$ and a tuple $\tz$ enumerating $\Vv \setminus X$
(in particular, $\tz$ is a subtuple of $\tv$), we put
\[
    (\QE_X\phi)^*:= \E \tz (\TM\tv \land \phi^*)
    \quad\text{ and }
    \quad (\QD_X\phi)^*:= \A \tz (\TM\tv \ra \phi^*).
\]
It remains to provide translations for the local team atoms $\beta(\tx)\in\Omega[\Vv]$.
The translation preserves their free variables.
For ease of notation, we explicitly give the translations just for atoms with two variables
$v_i, v_j$ from $\Vv$. It is obvious that the translations generalise to arbitrary tuples.
Let $\tz,\tz'$ be distinct copies of $\tv = (v_1,\dots,v_n)$.
For $j\leq n$, let $\tz_{-j}$ denote the tuple obtained by omitting $z_j$, and
let $\tz[j\mapsto y]=(z_1,\dots,z_{j-1},y,z_{j+1},\dots,z_n)$. Then we set
\begin{align*}
    (v_i \in v_j)^*     &:= \E \tz (\TM\tz \land z_j=v_i),\\
    (D_{v_i} v_j)^*     &:= \A \tz_{-i}\A\tz'_{-i}((\TM\tz[i \mapsto v_i] \land \TM\tz'[i \mapsto v_i]) \to z_j = z'_j),\\
    (\Ind_{v_i} v_j)^*  &:= \A \tz (\TM\tz \ra \E \tz' (\TM\tz' \land z'_i= v_i \land z'_j=z_j)).
\end{align*}
For the negations of these atoms, we consider the negation of their translation,
so for example $(v_i \notin v_j)^* = \lnot(v_i \in v_j)^*$.

Thus, the local logics of dependence and independence $\L[\Omega]$ can all be considered
as fragments of first-order logic. We remark that this restricted expressive power, compared to
the global logics of dependence and independence which are fragments of existential second-order logic, 
is not just due to the localisation of the team atoms.
Indeed, as we have seen, the global dependency atoms are easily expressible in the local logics,
via a further universal quantification with $\QGA$. But the localisation of the dependencies,
together with the global restriction of the available assignments to a \emph{fixed team},
permits to evaluate these logics by first-order rules,
as opposed to the inherent second-order operations in team semantics
such as the decomposition of teams, and their Skolem extensions.
In the next sections, we explore whether this more limited expressiveness is balanced by the benefits
of better algorithmic manageability and convenient model-theoretic properties.

\section{Decidability and undecidability}

We now want to discuss the question which of the logics $\L[\Omega]$ are decidable for
satisfiability. An obvious road for proving this is to show that the standard translation puts $\L[\Omega]$ 
into a known decidable fragment of $\FO$. This works in some cases, but not always, and
in particular, $\LFD$ seems to resist such an approach.
By different methods, Baltag and van Benthem have shown that $\LFD$ is decidable, and have
formulated the corresponding question for $\LFD^=$ as an open problem. We shall prove below that
$\LFD^=$ is undecidable.

\subsection{Embeddings into the guarded fragment}

A known decidable fragment of $\FO$ is the guarded fragment $\GF$.
In general, guarded logics arise as a natural generalization of modal logics.
Consider the standard translation $\phi\mapsto \phi^*(x)$
that identifies propositional modal logic $\ML$ 
with the modal fragment of $\FO$, by rewriting modal operators as relativised
quantifiers:
\[
    (\square \phi)^* = \forall y (Exy \impl \phi^*(y))
    \qquad\text{and}\qquad
    (\lozenge \phi)^* = \exists y(Exy \land \phi^*)).
\]
The guarded fragment $\GF$, introduced in \cite{AndrekaBenNem98}, 
generalises this idea to a much more powerful setting of first-order logic with 
an arbitrary relational vocabulary and an arbitrary number of variables.  
It lifts all restrictions of the modal fragment, except for the 
requirement that all quantifiers must be relativised (guarded) by some atom
that contains all free variables of the quantified formula.
More formally, $\GF$ is the smallest fragment of relational $\FO$ generated
from atomic formulae by propositional connectives and guarded quantification:
if  $\phi(\tx,\ty)$ is a formula in $\GF$ and $\alpha(\tx,\ty)$ is an atomic formula
that contains all free variables of $\phi(\tx,\ty)$,  then also 
\[
    \forall \ty(\alpha(\tx,\ty) \impl \phi(\tx,\ty)) 
    \qquad\text{and}\qquad
    \exists \ty(\alpha(\tx,\ty) \land \phi(\tx,\ty)) 
\]
are formulae of $\GF$.

The guarded fragment preserves, and to some degree explains, many of the good
model-theoretic and algorithmic properties of modal logics.
In particular, $\GF$ is decidable \cite{AndrekaBenNem98} and indeed has the finite model property
\cite{Graedel99a}: every satisfiable formula of $\GF$ has a finite model.

If we consider the standard translation of the logics $\L[\Omega]$ into $\FO$, we see that
that the translations of the logical operators preserve guardedness, so the questions is
just, which of the local atoms $\beta(\tx)$ can be rewritten as a guarded formula.
Our previous first-order translation of the local inclusion and exclusion atoms
can be rewritten as
\[
    (v_i \in v_j)^*:= \E \tz_{-j} \TM\tz[j\mapsto v_i]
    \quad\text{ and }\quad
    (v_i \not\in v_j)^*:= \A \tz_{-j} (\TM\tz[j\mapsto v_i]\ra \bot)
\] 
which are guarded formulae.

\begin{proposition}
    The local variant of inclusion-exclusion logic $\L[=,\neq,\in,\not\in]$ 
    is a fragment of $\GF$. In particular, it has the finite model property
    (and is therefore decidable for satisfiability and validity).
\end{proposition}

Recall that the global inclusion-exclusion logic, instead,
has the full power of independence logic and $\Sigma^1_1$.
The standard translations of local dependence, local anonymity,
and local independence are, however, not guarded.
For local independence logic, a straightforward argument shows that even without equality,
one can enforce cartesian products within the assignment space,
leading to the expressive power of usual first-order logic.
Hence $\L[\Ind]$ is already undecidable, and can therefore not be embedded into $\GF$.
The relationship of $\LFD$ with the guarded fragment,
or other guarded logics, is more difficult to analyse; 
we shall address this issue in \cref{sec:comp:gf}.

\subsection{Other decidability arguments}

\begin{proposition}
    $\L[D,\anon,\neq,\notin]$ is decidable for satisfiability.
\end{proposition}
\begin{proof}
    Given a dependence model $\M$ we denote by $\N$
    the corresponding $\LFD$-equivalent variable-distinguished (v.d.) model,
    as described in \cref{fact:all-dist}.
    We first prove via induction on $\phi \in \L[D,\anon,\neq,\notin]$ that
    whenever $\M \models_\TMS \phi$, we also have $\N \models_{\TNS} \phi$.

    If $\phi \in \{\tx \neq \ty,\ \tx \notin \ty\}$ and $\M \models_\TMS \phi$, then
    we know that $\tx \neq \ty$ \emph{as tuples of variables},
    so obviously $\N \models_{\TNS} \phi$.
    If otherwise $\phi \in L[D,\anon] \equiv \LFD$,
    this follows from $\M,\TMS \equiv_{\LFD} \N,\TNS$.
    The induction step for boolean connectives $\land$ and $\lor$ is clear.
    Also, note that if $\TMS,\TMT \in \TM$, then we have
    $\TMS(\tx) = \TMT(\tx)$ if and only if $\TNS(\tx) = \TNT(\tx)$.
    Hence the induction steps for the quantifiers $\QE_{\tx}\phi$ and $\QD_{\tx}\phi$
    are straightforward:
    If $\M \models_\TMS \QD_X\phi$ and $\TNT \in \TN$ with $\TNT =_X \TNS$,
    then also $\TMT =_X \TMS$, so $\M \models_\TMT \phi$. By induction hypothesis we obtain
    $\N \models_{\TNT} \phi$. Since $\TNT \in \TN$ with $\TNT =_X \TNS$ was arbitrary,
    we obtain $\N \models_\TNS \QD_X\phi$. The induction step for $\QE_X$ is analogous.
    This concludes the induction.

    In the following, let $[\tx \neq \ty] = \top$ if $\tx \neq \ty$ (as tuples of variables),
    and otherwise $[\tx \neq \ty] = \bot$.
    Then, since $\N$ is v.d., we have
    \[
        \N \models_{\TNS} \tx \neq \ty
        \quad\Iff\quad \N \models_{\TNS} [\tx \neq \ty].
    \]
    An analogue holds for $\tx \notin \ty$, so that over the class of v.d.~dependence models,
    we have $\tx \neq \ty \equiv \tx \notin \ty \equiv [\tx \neq \ty]$,
    and thereby $\L[D,\anon,\neq,\notin] \equiv \L[D,\anon] \equiv \LFD$.
    Finally, since the v.d.~model is $\LFD$-equivalent to the original,
    we see that $\LFD$-formulae are satisfiable if and only if they have a v.d.~model.
    Given $\phi \in \L[D,\anon,\neq,\notin]$, let $\phi'$ denote the $\LFD$-formula
    obtained from $\phi$ by replacing all occurrences of $\tx \neq \ty$ or $\tx \notin \ty$.
    Then
    \begin{align*}
        \phi\text{ satisfiable} &\quad\Iff\quad \M \models_\TMS \phi\text{ for some pointed dependence model }\M,s\\
                                &\quad\Iff\quad \N \models_{\TNS} \phi\\
                                &\quad\Iff\quad \N \models_{\TNS} \phi'\\
                                &\quad\Iff\quad \phi'\text{ satisfiable}.
    \end{align*}
    Thus we have constructed a satisfiability-preserving reduction from
    $\L[D,\anon,\neq,\notin]$ to $\LFD$.
    Since $\LFD$ is known to be decidable for satisfiability \cite{BaltagBen20}, this concludes the proof.
\end{proof}

\subsection{Undecidability: dependence together with equality or inclusion}

We now solve a main open problem from \cite{BaltagBen20} by proving that
$\LFD^=$ is undecidable.
In fact, we will prove that $\L[D,\in]$ is undecidable
and then show how to adapt the argument to $\L[D,=]$,
which is contained in $\L[D,\anon,=,\neq] \equiv \LFD^=$.
The crucial reason for the undecidability is that,
in the presence of either equality or inclusion,
no analogue of \cref{fact:all-dist} can be established,
and we can use these atoms to copy values between variables,
while keeping certain other values fixed.
We demonstrate this idea for inclusion:
\begin{example}
    Consider a dependence model $\M$ over variables $\Vv = \{x,y,z\}$
    such that $\M \models xyx \subseteq xyz$, i.e.~$\M \models \QGA(xyx \in xyz)$,
    meaning $\M \models_\TMS xyx \in xyz$ for all $\TMS \in \TM$.
    Then, given some $\TMS \in \TM$ with $\TMS(x,y,z) = (a,b,c)$, there exists $\TMT \in \TM$
    such that $\TMT(x,y,z) = (a,b,a)$. In this sense, we copied the value $a$ of $x$
    to the variable $z$, while keeping the value of the variables $x$ and $y$ fixed.
\end{example}

We recall some basic notions concerning the classical decision
problem for first-order logic.
For details we refer  to \cite[Chapter 3.1]{BoergerGraGur97}.
A syntactic fragment $F\subseteq\FO$ is called a
\emph{conservative reduction class} if there exists a conservative
reduction $f:\FO\ra F$, i.e. a computable function
that preserves satisfiability and finite satisfiability in both directions.
If $F$ is a conservative reduction class, then 
by Trakhtenbrot's Theorem, the satisfiability and finite satisfiability problems for 
$F$ are undecidable.

One of the classical conservative reduction classes is
the Kahr-Class, denoted $[\forall\exists\forall, (\omega,1)]$,
which consists of the sentences of form $\forall x \exists y \forall z \phi(x,y,z)$
where $\phi$ is quantifier-free, without equality, and
which may only use a single binary relation but an unbounded number of monadic ones.

We shall construct a conservative reduction from the Kahr-Class
into fragments of $\L[D,\in]$ and $\L[D,=]$.
We require only: four variables,
one binary predicate and an unbounded number of monadic ones,
one dependence atom (used positively), and six inclusions/equalities (also used positively).

\begin{theorem}\label{thm:cons-red}
    The four-variable fragments of $\L[=,D]$ and $\L[D,\in]$ are conservative reduction classes.
    In particular, the satisfiability, validity and finite satisfiability problems for $\LFD^=$ are undecidable.
\end{theorem}

\begin{proof}
 We first prove the claim for $\L[D,\in]$, and then show how to adapt the proof
 to $\L[D,=]$. Since the Kahr-Class is a conservative reduction class, it suffices,
 by transitivity, to exhibit a conservative reduction $ \psi \mapsto \psi^*$
from $[\forall\exists\forall, (\omega,1)]$  to $\L[D,\in]$.

Notice that, for  $\psi = \forall x\exists y \forall z \phi(x,y,z)$ in the Kahr-Class, 
the quantifier-free part  $\phi$ is in the base logic $\L$. We define
\[
    \psi^* \coloneqq \QGA\phi(x,y,z) \land \dep(x,y) \land \bigwedge_{i=0}^5 \vartheta_i.
\]

The subformulae $\vartheta_i$ allow us to copy the value of a variable
to another, \emph{while keeping certain other variables fixed},
as demonstrated in the example above.
This is used to enforce a cartesian product
included in $\TM(x,z)$ for the dependence models $\M$,
allowing us to retrieve a classical model from them.
We require an extra variable $v$ as an additional
``temporary storage'' to copy values between variables.
So while we keep $\phi$ to contain only the variables $x,y,z$,
overall we will need four variables, whose order we fix to $x,y,z,v$.
The $\vartheta_i$ are defined as:
\begin{alignat*}{2}
    &\vartheta_0 \coloneqq xx  \subseteq xz  \qquad\qquad&&(\text{copy $x$ to $z$, keeping $x$}),\\
    &\vartheta_1 \coloneqq yyv \subseteq yzv \qquad\qquad&&(\text{copy $y$ to $z$, keeping $y,v$}),\\
    &\vartheta_2 \coloneqq xzx \subseteq xzv \qquad\qquad&&(\text{copy $x$ to $v$, keeping $x,z$}),\\
    &\vartheta_3 \coloneqq yzy \subseteq yzv \qquad\qquad&&(\text{copy $y$ to $v$, keeping $y,z$}),\\
    &\vartheta_4 \coloneqq zzv \subseteq xzv \qquad\qquad&&(\text{copy $z$ to $x$, keeping $z,v$}),\\
    &\vartheta_5 \coloneqq vzv \subseteq xzv \qquad\qquad&&(\text{copy $v$ to $x$, keeping $z,v$}).
\end{alignat*}
Since $\dep(x,y) \equiv \QGA D_xy$ and analogously
$\tx \subseteq \ty \equiv \QGA(\tx \in \ty)$ it follows that $\psi^*$ indeed
is in the four-variable fragment of $\L[D,\in]$.

We have to prove the following two claims:
\begin{enumerate}
    \item A (finite) model of $\psi$ induces a (finite) dependence model of $\psi^*$.\label[claim]{clm:cons-red:1}
    \item A (finite) dependence model of $\psi^*$ induces a (finite) model of $\psi$.\label[claim]{clm:cons-red:2}
\end{enumerate}

To prove \cref{clm:cons-red:1},
assume that we have a model $\frA\models\psi$ with universe $A$.
Thus there exists a function $f \colon A \to A$ such that
$\frA \models \phi(a,fa,b)$ for all $a,b \in A$.
We construct the dependence model $(\frA,\TM)$ with team $\TM$ given by
\[
    \TM \coloneqq \{(a,fa,b,c) \mid a,b,c \in A\}.
\]
Remember that we denote assignments by their tuple of values,
so $(a,fa,b,c)$ represents the assignment $(x,y,z,v) \mapsto (a,fa,b,c)$.
It is clear that $y$ globally depends on $x$, i.e.~$(\frA,\TM) \models \dep(x,y)$,
and that by the choice of $\TM$ we also have $(\frA,\TM) \models \QGA\phi(x,y,z)$.
The $\vartheta_i$ are satisfied in $\M$, since $\TM(x,z,v) = A^3$
is a cartesian product of the whole universe.
Overall, we obtain $(\frA,\TM) \models \psi^*$.
Clearly, if $\frA$ is a finite model, then $(\frA,\TM)$ is finite as well.
This completes the proof of \cref{clm:cons-red:1}.

For the converse, \cref{clm:cons-red:2},
suppose that we have a dependence model $\M$ such that $\M \models \psi^*$.
Because of the global dependence $\M \models \dep(x,y)$ there exists a function
$f \colon \TM(x) \to \TM(y)$ such that $\TMT(y) = f(\TMT(x))$ for all $\TMT \in \TM$.
Note that $\TM(y) \subseteq \TM(x)$, since $\vartheta_1$ and $\vartheta_4$
allow us to copy values from $y$ to $z$ and from there to $x$.
Hence we have $f \colon \TM(x) \to \TM(x)$, i.e.~we can iterate $f$ on values of $x$.
Fix some arbitrary $\TMS \in \TM$ and set $\und{0} \coloneqq \TMS(x)$, as well as
$\und{i} \coloneqq f^i\und{0}$ for $i \in \mathbb{N}$.

We construct a model for $\psi$ by $\frA \coloneqq \frM \restr A$
where $A \coloneqq \{\und{i} \mid i \in \bbN\}$.
The function $f \restr A \colon A \to A$ plays the role of the Skolem function
for $y$ in the quantification $\forall x \exists y \forall z$ of $\psi$.
We need to ensure that $\phi(a,fa,b)$ holds in $\frM$ (and thus in $\frA$)
for all $a,b \in A$.
    
Since $\M \models \QGA\phi(x,y,z)$ we know that $\frM \models \phi(\TMT(x),f(\TMT(x)),\TMT(z))$
for all $\TMT\in \TM$.
Hence it suffices to show that
\begin{equation}\label{eq:univ}
    A \times A \subseteq \TM(x,z).
\end{equation}
In the following we write $*$ as placeholder for not further specified elements of $\frM$.
The expression $\TMT \xrightarrow{\vartheta_i} \TMU$ for some $\TMT \in \TM$
denotes that the existence of
$\TMU \in \TM$ follows by applying the ``copy-rule'' which $\vartheta_i$ represents.
Notice that for all $\TMT \in \TM$ with $\TMT(x) = \und{i}$ we have
$\TMT(y) = f(\und{i}) = \und{i+1}$.
In particular, keeping the value of $x$ implies keeping the value of $y$.
\begin{enumerate}
    \item $(\und{0},\und{0}) \in \TM(x,z)$:\\
        We know that $\TMS$ looks like $(\und{0},\und{1},*,*)$.
        Since $\TMS \in \TM$ and
        \[
            \TMS = (\und{0},\und{1},*,*)
            \xrightarrow{\vartheta_0} (\und{0},\und{1},\und{0},*)
            \eqqcolon \TMT\qquad\qquad\text{(copy $x$ to $z$, keeping $x,y$)}
        \]
        we see that $\TMT \in \TM$ with $\TMT(x,z) = (\und{0},\und{0})$.

    \item If $(\und{0},\und{j}) \in \TM(x,z)$, then also $(\und{0},\und{j+1}) \in \TM(x,z)$:\\
        By assumption we have $\TMT = (\und{0},\und{1},\und{j},*) \in \TM$.
        Together with the derivation
        \begin{alignat*}{2}
            \TMT =\ &(\und{0},\und{1},\und{j},*)\\
            \xrightarrow{\vartheta_2}\quad &(\und{0},\und{1},\und{j},\und{0})
                && (\text{copy $x$ to $v$, keeping $x,y,z$})\\
            \xrightarrow{\vartheta_4}\quad &(\und{j},\und{j+1},\und{j},\und{0})
                &&(\text{copy $z$ to $x$, keeping $z,v$})\\
            \xrightarrow{\vartheta_1}\quad &(*,\und{j+1},\und{j+1},\und{0})
                &&(\text{copy $y$ to $z$, keeping $y,v$})\\
            \xrightarrow{\vartheta_5}\quad &(\und{0},\und{1},\und{j+1},\und{0}) \eqqcolon \TMU
            \qquad\qquad&&(\text{copy $v$ to $x$, keeping $z,v$})
        \end{alignat*}
        we obtain $\TMU \in \TM$ with $\TMU(x,z) = (\und{0},\und{j+1})$.

    \item If $(\und{i},\und{j}) \in \TM(x,z)$ then also $(\und{i+1},\und{j}) \in \TM(x,z)$:\\
        By assumption we have $\TMT = (\und{i},\und{i+1},\und{j},*) \in \TM$.
        Together with the derivation
        \begin{alignat*}{2}
            \TMT =\ &(\und{i},\und{i+1},\und{j},*)\\
            \xrightarrow{\vartheta_3}\quad &(*,\und{i+1},\und{j},\und{i+1})
                && (\text{copy $y$ to $v$, keeping $y,z$})\\
            \xrightarrow{\vartheta_5}\quad &(\und{i+1},\und{i+2},\und{j},\und{i+1}) \eqqcolon \TMU
                \qquad\qquad&& (\text{copy $v$ to $x$, keeping $z,v$})
        \end{alignat*}
        we obtain $\TMU \in \TM$ with $\TMU(x,z) = (\und{i+1},\und{j})$.
\end{enumerate}
Now the inclusion (\ref{eq:univ}) follows via induction.
By the above argument this proves that $\frA \models \phi(a,fa,b)$
for all $a,b \in A$ and hence $\frA \models \forall x \exists y \forall z \phi(x,y,z)$.
Again it is clear that if $\M$ is finite, then so is $\frA$.
This concludes the proof of \cref{clm:cons-red:2},
showing that the four-variable fragment of $\L[D,\in]$ is a conservative reduction class.

For the case of $\L[D,=]$, note first that for any dependence model $\M$
with $\TMS \in \TM$ and non-empty tuples of variables $\tx,\ty,\tz$ with $[\ty] \subseteq [\tx]$, we have
\[
    \M \models_\TMS \tx\ty \in \tx\tz
    \quad\Iff\quad
    \M \models_\TMS \QE_{\tx\ty}(\ty = \tz)
    \quad\Iff\quad
    \M \models_\TMS \QE_{\tx}(\ty = \tz).
\]
Thus, if $[\ty] \subseteq [\tx]$,
we have $\tx\ty \subseteq \tx\tz \equiv \QGA\QE_{\tx}(\ty = \tz)$.
In our case, all inclusion atoms can be written in this form.
Hence, we can define equivalent formulae in $\L[D,=]$:
\begin{alignat*}{2}
    &\vartheta_0 \coloneqq \QGA\QE_{x}(x=z)\qquad\qquad&&(\text{copy $x$ to $z$, keeping $x$}),\\
    &\vartheta_1 \coloneqq \QGA\QE_{yv}(y=z)   &&(\text{copy $y$ to $z$, keeping $y,v$}),\\
    &\vartheta_2 \coloneqq \QGA\QE_{xz}(x=v)   &&(\text{copy $x$ to $v$, keeping $x,z$}),\\
    &\vartheta_3 \coloneqq \QGA\QE_{yz}(y=v)   &&(\text{copy $y$ to $v$, keeping $y,z$}),\\
    &\vartheta_4 \coloneqq \QGA\QE_{zv}(z=x)    &&(\text{copy $z$ to $x$, keeping $z,v$}),\\
    &\vartheta_5 \coloneqq \QGA\QE_{zv}(v=x)    &&(\text{copy $v$ to $x$, keeping $z,v$}).
\end{alignat*}
The proof then works in the exact same way as before.
\end{proof}

\subsection{Classification by Decidability: Conclusion}

Consider the lattice of the local logics using the atoms $D,\anon,=,\neq,\in,\notin,\Ind$.
Because we often want $\Omega$ to be closed under negation, we may also consider $\lnot\Ind$.
In the above discussions, we have shown that $\L[D,\in],\L[D,=]$ and $\L[\Ind]$
are minimal undecidable extensions of $\L$ in this lattice.
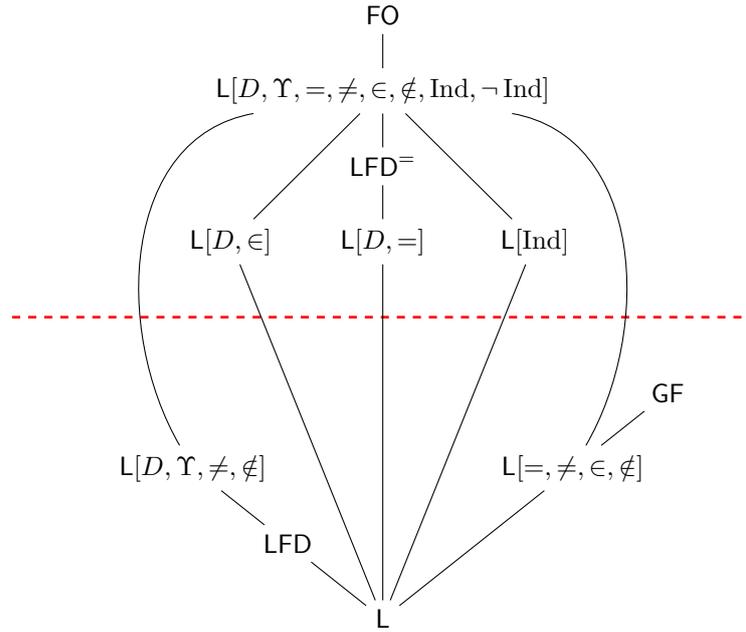
\begin{center}
\begin{tikzpicture}
    \node (tmax) at (0,5) {$\FO$};
    \node (max) at (0,4) {$\L[D,\anon,=,\neq,\in,\notin,\Ind,\lnot\Ind]$};
    \node (a) at (-2,2) {$\L[D,\in]$};
    \node (aa) at (0,3) {$\LFD^=$};
    \node (b) at (0,2) {$\L[D,=]$};
    \node (c) at (2,2) {$\L[\Ind]$};
    \node (h) at (3.75, 0) {$\GF$};
    \node (f) at (2.5,-1) {$\L[=,\neq,\in,\notin]$};
    \node (d) at (-2.5,-1) {$\L[D,\anon,\neq,\notin]$};
    \node (g) at (-1.25,-2) {$\LFD$};
    \node (min) at (0,-3) {$\L$};
    \node (left) at (-5,1) {};
    \node (right) at (5,1) {};
    \draw (tmax) -- (max) -- (aa) -- (b);
    \draw (max) -- (a);
    \draw (max) -- (c);
    \draw (min) -- (g) -- (d);
    \draw (min) -- (f) -- (h);
    \draw[dashed, color=red,line width=1pt] (left) -- (right);
    \draw (d) to [out=120, in=190] (max);
    \draw (f) to [out=60, in=-10] (max);
    \draw (min) -- (a);
    \draw (min) -- (b);
    \draw (min) -- (c);
\end{tikzpicture}
\end{center}
With this, many of the extensions of $\L$ by local atoms
$\Omega \subseteq \{D,\anon,=,\neq,\in,\notin,\Ind,\lnot\Ind\}$
have been classified by decidability.
The main question left is how negated local independence $\lnot \Ind$ fits into this picture,
and how anonymity and inclusion / equality affect each other.

\begin{open}
    Classify the remaining $\L[\Omega]$ where
    $\Omega \subseteq \{D,\anon,=,\neq,\in,\notin,\Ind,\lnot\Ind\}$
    is not a subset of $\{D,\anon,\neq,\notin\}$
    or $\{=,\neq,\in,\notin\}$ and not a superset of $\{D,\in\},\{D,=\}$ or $\{\Ind\}$.
    Of the extensions fitting this description $\L[\anon,=],\L[\anon,D]$ and $\L[\lnot\Ind]$ are minimal,
    whereas $\L[D,\anon,\neq,\notin,\lnot\Ind]$ and $\L[\anon,=,\neq,\in,\notin,\lnot\Ind]$ are maximal.
\end{open}

\section{Bisimulation and Ehrenfeucht-Fra\"issé theory for LFD}

We define a notion of bisimulation for local dependence logics $\L[\Omega]$
and in particular for $\LFD$ and $\LFD^=$
in such a way that we obtain an analogue of the classical
Ehrenfeucht-Fra\"issé Theorem in \cref{sec:bisim:ef},
and later an analogue of van Benthems's Theorem in \cref{sec:comp:fo:char}.
In the following, many results will require that $\Omega$ and hence $\L[\Omega]$
is closed under negation.
We shall also consider infinitary variants of these  logics.

\subsection{Bisimulation}

\begin{Definition}[Bisimulation]\label{def:bisim}
    Let $\M$ and $\N$ be two dependence models of the same
    type $(\tau,\Vv)$. A binary relation $Z \subseteq \TM \times \TN$
    is an $\L[\Omega]$-bisimulation between $\M$ and $\N$ if for all $(\TMS,\TNS) \in Z$:
    \begin{enumerate}
        \item $\M,\TMS$ and $\N,\TNS$ agree on the atoms of $\L[\Omega]$:
            \begin{enumerate}
                \item For all $R \in \tau$ and $\tx \in \Vv^{\ar(R)}$ we have
                    $\M \models_\TMS R\tx\ \Iff\ \N \models_\TNS R\tx$.
                \item For all local atoms $\beta\in \Omega[\Vv]$
                    we have $\M \models_\TMS \beta\ \Iff\ \N \models_\TNS \beta$.
            \end{enumerate}
        \item (back) For all $\TNT \in \TN$ and all \emph{finite}
            $X \subseteq \comvar{\TNT}{\TNS} \coloneqq \{x \in \Vv \mid \TNT =_x \TNS\}$
            there is some $\TMT \in \TM$ with $(\TMT,\TNT) \in Z$ and $\TMT =_X \TMS$.
        \item (forth) For all $\TMT \in \TM$ and all \emph{finite}
            $X \subseteq \comvar{\TMT}{\TMS} \coloneqq \{x \in \Vv \mid \TMT =_x \TMS\}$
            there is some $\TNT \in \TN$ with $(\TMT,\TNT) \in Z$ and $\TNT =_X \TNS$.
    \end{enumerate}
\end{Definition}

We restrict ourselves to finite sets $X \subseteq \comvar{\TMT}{\TMS}$
because $\L$ only allows finite sets within our modalities $\QD_X$ and $\QE_X$.
Whenever $\M, \N$ and $\L[\Omega]$ are clear from context, we write
$\TMS \sim \TNS$ if there exists an $\L[\Omega]$-bisimulation
$Z$ between $\M$ and $\N$ with $(\TMS,\TNS) \in Z$.

\begin{Definition}[Ordinal approximations to bisimulation]\label{def:bisim-approx}
    We write $\TMS \sim^0 \TNS$ and say that $\TMS$ and $\TNS$ are $\alpha$-bisimilar
    if $\M,\TMS$ and $\N,\TNS$ agree on $\L[\Omega]$-atoms.
    Now define $\TMS \sim^{\alpha} \TNS$ for ordinals $\alpha \in \On$ by induction;
    when defining $\TMS \sim^{\alpha+1} \TNS$,
    we require the conditions
    \begin{itemize}
        \item $(\alpha+1)$-back:
            For all $\TNT \in \TN$ and all \emph{finite}
            $X \subseteq \comvar{\TNT}{\TNS}$ there exists some $\TMT \in \TM$
            with $\TMT \sim^{\alpha} \TNT$ and $\TMT =_X \TMS$.
        \item $(\alpha+1)$-forth:
            For all $\TMT \in \TM$ and all \emph{finite}
            $X \subseteq \comvar{\TMT}{\TMS}$ there exists some $\TNT \in \TN$
            with $\TMT \sim^{\alpha} \TNT$ and $\TNT =_X \TNS$.
    \end{itemize}
    For limit ordinals $\lambda$,
    we say that $\TMS \sim^{\lambda} \TNS$ if
    $\TMS \sim^{\alpha} \TNS$ for all $\alpha < \lambda$,
    so essentially
    \begin{equation}\label{eq:approx-bisim-lim-ord}
        {\sim^\lambda} = \bigcap_{\alpha < \lambda} {\sim^\alpha}.
    \end{equation}
    As usual, $\alpha$-bisimilarity implies
    $\beta$-bisimilarity for all $\beta < \alpha$.
    Furthermore, full bisimilarity is now simply given by
    ${\sim} = \bigcap_{\alpha \in \On} {\sim^\alpha}$.
\end{Definition}

We want to emphasize that the back and forth
conditions of our bisimulations do not require the regarded assignments $\TMT$ or $\TNT$
to actually agree on any variables with $\TMS$ or $\TNS$ respectively,
i.e.~$\comvar{\TMT}{\TMS}$ and $\comvar{\TNT}{\TNS}$ may be empty.
The reason for this is that $\L$ has the global modalities
$\QGA = \QD_{\emptyset}$ and $\QGE = \QE_{\emptyset}$
and that we want bisimilarity to correspond to logical equivalence
(we say that some modality is global if its corresponding accessibility relation is the all-relation,
i.e.~contains all possible pairs of objects, as is the case for $=_\emptyset$ on teams).
As an example, let $\Vv$ be finite and
\[
    \psi \coloneqq \QGE(Rx) \land \bigwedge_{v \in \Vv} \lnot \QE_{v}Rx.
\]
Then $\M \models_\TMS \psi$ means that there is some
$\TMT \in \TM$ with $\M \models_{\TMT} Rx$,
but that there is no $\TMU \in \TM$ with $u =_x s$ having this property,
and thus $\comvar{\TMT}{\TMS} = \emptyset$.
Since $\L$ is able to witness this $\TMT$, it is natural
to require a bisimilar $\TNT \in \TN$ for \emph{every} $\TMT \in \TM$,
and not just for those $\TMT$ that agree with $\TMS$ on some variable.
Likewise for the back condition.

This forces every bisimulation to be \emph{global},
meaning that every assignment in $\TM$
is bisimilar to at least one assignment in $\TN$, and vice versa.
This is a common consequence of having global modalities;
in the context of, say, $\ML$ with an explicitly added global modality,
often denoted $\ML(\forall)$, the canonical bisimulation is just the global version of
ordinary $\ML$-bisimulation \cite{DawarOtto09}.

Since we already defined bisimulation for infinite ordinals,
the corresponding step for our logics is to consider their infinitary variants.
We briefly give some routine definitions which we need in what follows.

\begin{Definition}[{$\L[\Omega]_\infty$}]
    The infinitary extension $\L[\Omega]_\infty$ of $\L[\Omega]$ allows
    conjunction and disjunction over arbitrarily large sets of $\L[\Omega]_\infty$-formulae,
    with the obvious semantics.
\end{Definition}

\begin{Definition}[Quantifier Rank]
    The quantifier rank $\qr(\phi)$ of some $\phi \in \L[\Omega]_\infty$ is a
    recursively defined ordinal. We define $\qr(\beta) = 0$ for atoms $\beta$
    (including negated relational atoms and the local atoms),
    as well as $\qr(\QD_X\phi) = \qr(\QE_X\phi) = \qr(\phi) +1$.
    Lastly, we set
    $\qr\left(\bigwedge_{i \in I} \phi_i\right)
    = \qr\left(\bigvee_{i \in I} \phi_i\right)
    = \sup_{i \in I} \qr(\phi_i)$.
\end{Definition}

\begin{Definition}[Equivalence]~
    We use the usual symbols $\equiv_{\L[\Omega]}$ ($\equiv_{\L[\Omega]}^{\alpha}$)
    for equivalence of pointed dependence models in the logic $\L[\Omega]$
    (up to quantifier rank $\alpha$).
    We also use the short form $\TMS \equiv_{\L[\Omega]}^\alpha \TNS$ whenever $\M,\TMS$ and $\N,\TNS$
    are clear from context.
    The infinitary case ${\L[\Omega]}_\infty$ is defined analogously, but we often
    write $\equiv_{\L[\Omega]}^\infty$ instead of $\equiv_{\L[\Omega]_\infty}$.
\end{Definition}

\subsection{An Ehrenfeucht-Fra\"issé Theorem}
\label{sec:bisim:ef}

Given a logic, a common goal is to find a correspondence
between logical indistinguishability and behavioural
equivalence in some structural form,
often as a relation akin to bisimulation, a collection of partial isomorphisms,
or a winning strategy of certain two-player games.
For $\FO$ we have back-and-forth systems and Ehrenfeucht-Fra\"issé games (cf.~\cite[Chapter 3.3]{Hodges97a}),
whereas for $\ML$ one has ordinary bisimulation
and the corresponding bisimulation games (cf.~\cite[Chapter 2.2]{BlackburnRijVen01}).
The following results and in particular \cref{thm:bisim}
show that the bisimulations defined above
fulfill such a role for $\L[\Omega]$ whenever $\Omega$ is closed under negation.

\begin{lemma}\label{lemma:char-form-fin}
    Let $\M,\TMS$ be a dependence model of some finite type $(\tau,\Vv)$ and $\Omega$ be finite
    and closed under negation.
    For every $k \in \bbN$ there exists a formula $\chi^k_\TMS \in \L[\Omega](\tau,\Vv)$
    of quantifier rank $k$ that defines the
    $\sim^k$-class of $\TMS$, so that for all suitable $\N,\TNS$
    \[
        \N \models_\TNS \chi^k_{\TMS}
        \quad\Iff\quad
        \TNS \sim^k \TMS
    \]
    Up to $\L[\Omega]$-equivalence,
    the number of such $\chi^k_{\TMS}$ is finite and depends only on $\tau,\Vv,\Omega$ and $k$.
\end{lemma}
\begin{proof}
    Since $\tau,\Vv$ and $\Omega$ are finite,
    there are up to equivalence only finitely many formulae
    $\phi \in \L[\Omega](\tau,\Vv)$ with $\qr(\phi) = 0$,
    which allows us to define
    \[
        \chi^0_{\TMS} \coloneqq \bigwedge\{\phi \in \L[\Omega](\tau,\Vv) \mid \qr(\phi) = 0,\ \M \models_\TMS \phi\}.
    \]
    One proceeds inductively by defining
    $\chi_{\TMS}^{k+1} \coloneqq \phi^{k+1}_{\text{back}} \land \phi^{k+1}_{\text{forth}}$, where
    \[
        \phi_{\text{back}}^{k+1} \coloneqq
            \bigwedge_{X \subseteq \Vv}
            \QD_X \bigvee_{\substack{\TMT \in \TM\\ \TMT =_X \TMS}}
            \chi^k_{\TMT}
        \quad\qquad\text{and}\qquad\quad
        \phi_{\text{forth}}^{k+1} \coloneqq
            \bigwedge_{\TMT \in \TM}\
            \bigwedge_{\substack{X\subseteq \Vv\\[1pt] \TMT =_X \TMS}}
            \QE_{X}\chi^k_{\TMT}.
    \]
    By the induction hypothesis, these are well-defined $\L[\Omega](\tau,\Vv)$-formulae
    and correspond precisely to the $(k+1)$-back and $(k+1)$-forth
    conditions as given in \cref{def:bisim-approx}.
\end{proof}

\begin{lemma}\label{lemma:char-form-inf}
    Let $\M,\TMS$ be a dependence model of some type $(\tau,\Vv)$ and $\Omega$ be closed under negation.
    For every $\alpha \in \On$ there exists a formula $\chi^\alpha_\TMS \in \L[\Omega](\tau,\Vv)$
    of quantifier rank $\alpha$ that defines the
    $\sim^\alpha$-class of $\TMS$, so that for all suitable $\N,\TNS$
    \[
        \N \models_\TNS \chi^\alpha_{\TMS}
        \quad\Iff\quad
        \TNS \sim^\alpha \TMS.
    \]
\end{lemma}
\begin{proof}
    For $\alpha = 0$ and successor ordinals,
    we use analogous definitions to the ones above,
    except that we have to explicitly consider only \emph{finite} $X \subseteq \Vv$
    as they are used within the quantifiers $\QD_X$ and $\QE_X$.
    For limit ordinals $\lambda$, we set
    $
        \chi^{\lambda}_\TMS \coloneqq
        \bigwedge_{\alpha < \lambda} \chi^{\alpha}_\TMS,
    $
    which corresponds to definition of $\sim^{\lambda}$,
    see \cref{eq:approx-bisim-lim-ord} in \cref{def:bisim-approx}.
    Conclude via transfinite induction.
\end{proof}

\begin{theorem}[{Ehrenfeucht-Fra\"issé and Karp theorems for $\L[\Omega]$}]\label{thm:bisim}~\\
    Let $\Omega$ be closed under negation.
    If $\Omega$ is \emph{finite} and
    $\M,\TMS$ and $\N,\TNS$ are dependence models of the same \emph{finite} type, then
    \[
        \TMS \sim_{\L[\Omega]}^k \TNS
        \quad\Iff\quad
        \TMS \equiv_{\L[\Omega]}^k \TNS,
        \qquad k \in \bbN.
    \]
    As a consequence we obtain that under those same conditions
    \[
        \TMS \sim_{\L[\Omega]}^{\omega} \TNS
        \quad\Iff\quad
        \TMS \equiv_{\L[\Omega]} \TNS.
    \]
    For $\Omega$ not necessarily finite and arbitrary types we obtain
    \[
        \TMS \sim_{\L[\Omega]}^\alpha \TNS
        \quad\Iff\quad
        \TMS \equiv_{\L[\Omega]_\infty}^\alpha \TNS,
        \qquad\alpha \in \On
    \]
    and therefore
    \[
        \TMS \sim_{\L[\Omega]} \TNS
        \quad\Iff\quad
        \TMS \equiv_{\L[\Omega]}^\infty \TNS.
    \]
\end{theorem}
\begin{proof}
    The proof is a routine induction.
    One shows that $k$-$\L[\Omega]$-bisimilarity entails $\L[\Omega]$-equivalence
    up to quantifier rank $k$ by using the characteristics of our bisimulation,
    whereas the converse implication is immediate from the above lemmas.
    The infinitary case is handled analogously.
\end{proof}

This allows us to show undefinability of
some property of (pointed) dependence models,
by finding two such models that are bisimilar but differ on said property.
By the above theorem, we then know that these models
are logically indistinguishable in $\L[\Omega]$,
so the considered property cannot be defined in the respective logic.

\begin{example}\label{ex:inc}
    Let $\Omega = \{D,\anon,=,\neq\}$.
    We show that the inclusion $\TM(x) \subseteq \TM(y)$ is not $\L[\Omega]_\infty$-definable.
    It suffices to show this for $\tau = \emptyset$ and $\Vv = \{x,y\}$,
    because the example below can be adapted accordingly.
    Consider dependence models $\M,\N$ of type $(\emptyset,\{x,y\})$
    with teams given by
    \[
        \TM \coloneqq \{(a,b),(b,a)\}
        \qquad\text{and}\qquad
        \TN \coloneqq \{(1,2),(2,0)\}.
    \]
    Note that $\TM(x) \subseteq \TM(y)$, but $\TN(x) \not\subseteq \TN(y)$.
    Now let $Z$ be the binary relation on
    $\TM \times \TN$ defined by
    $(a,b) \mathbin{Z} (1,2)$ and $(b,a) \mathbin{Z} (2,0)$.
    It is easy to verify that $Z$ is a full $\L[\Omega]$-bisimulation and
    hence $\TMS \equiv_{\L[\Omega]}^{\infty} \TNS$ by \cref{thm:bisim},
    for all $(\TMS,\TNS) \in Z$.
    Indeed, note that $\anon_\emptyset x,\anon_\emptyset y, D_xy,D_yx$
    and $x \neq y$ hold at all assignments in both teams,
    so the pairs of assignments agree on atoms.
    Furthermore, in both teams we see that the two assignments
    do not agree on any variables, so evidently the only choice we have
    at the back and forth clauses (every assignment stands in $Z$-relation to exactly one other)
    always works out.
\end{example}

\section{A Characterisation Theorem}\label{sec:comp:fo:char}

We now prove a characterisation theorem for $\L[\Omega]$.
The theorem is an analogue of van Benthem's Theorem, which states
that $\ML$, via its standard translation into $\FO$,
is precisely the bisimulation-invariant fragment of $\FO$ over the class of pointed Kripke structures.
It was first formulated in \cite{Benthem76} and \cite{Benthem83}.
We adapt a well known proof using saturated structures
by following the exposition in \cite[Chapter 2.6]{BlackburnRijVen01}.
We assume that the reader is familiar with basic model-theoretic notions
such as elementary extensions and $\omega$-saturated structures.

Given a dependence model $\M,\TMS$ of finite type $(\tau,\Vv)$,
we have already seen that we can interpret $\M$ as a $\tau\cup\{\TM\}$-structure,
where $\TM$ is a $|\Vv|$-ary relation symbol; this is done by fixing an enumeration $\tv$ of $\Vv$
and viewing $\TM$ as the $|\Vv|$-ary relation $\TM(\tv) = \{\TMS(\tv) \mid \TMS \in \TM\}$ over $\M$.
In the following, we also write $\TTMS \coloneqq \TMS(\tv)$ for the $|\Vv|$-ary tuple
corresponding to $\TMS \in \TM$. For clarity of presentation we denote the evaluation
of the first-order formula $\phi(\tx) \in \FO(\tau \cup \{\TM\})$ in the structure $\M$ at
the tuple $\TTMS$ under classical Tarski-semantics by $\M \models^{\FO} \phi(\TTMS)$.
In \cref{subsec:trans} we gave a first-order translation $\L[\Omega] \to \FO, \phi \mapsto \phi^*$ such that
for all $\phi \in\L[\Omega]$ and every fitting dependence model $\M,\TMS$ we have
\[
    \M \models_\TMS \phi \quad\Iff\quad \M \models^{\FO} \phi^*(\TTMS).
\]
In the following, let $(\tau,\Vv)$ be a finite type with an enumeration $\tv$ of $\Vv$,
and $\Omega \subseteq \{D,\anon,=,\neq,\in,\notin,\Ind,\cdots\}$
a finite set of local atoms that is closed under negation
and to which we may extend the above first-order translation.

\begin{Definition}\label{def:lfd-theory}
    For dependence models $\M,\TMS$  and $\N,\TNS$ let
    \[
        \Th_{\L[\Omega]}(\M,\TMS) \coloneqq
        \{\phi^* \mid \phi \in \L[\Omega],\ \M \models_\TMS \phi\} \subseteq \FO(\tau \cup \{\TM\}).
    \]
    We use $\Th_{\L[\Omega]}(\TMS)$ if $\M$ is clear from context.
    Note that $\Th_{\L[\Omega]}(\TMS)$ generally
    contains formulae with free variables among $\tv$, not just sentences.
\end{Definition}

\begin{lemma}
    Let $\M,\TMS$ and $\N,\TNS$ be pointed dependence models. Then
    \[
        \M \models^{\FO} \Th_{\L[\Omega]}(\TNS)(\TTMS)
        \quad\Iff\quad
        \Th_{\L[\Omega]}(\TMS) = \Th_{\L[\Omega]}(\TNS)
        \quad\Iff\quad
        \TMS \equiv_{\L[\Omega]} \TNS.
    \]
\end{lemma}

\begin{definition}\label{def:vartheta}
    Write $n = |\Vv| = |\tv|$ and consider tuples of variables $\tx = (x_1,\dots,x_n)$
    and $\ty = (y_1,\dots,y_n)$.
    With respect to the ordering of $\Vv$ given by $\tv$, we identify
    $X \subseteq \Vv$ with a set of indices $I_X \subseteq \{1,\dots,n\}$.
    Now define $\vartheta_X(\tx,\ty) \in \FO(\tau \cup \{\TM\})$ by
    \[
        \vartheta_X(\tx,\ty) \coloneqq \TM\tx \land \TM\ty \land \bigwedge_{i \in I_X} x_i = y_i.
    \]
\end{definition}

\begin{lemma}
    Let $\M$ be a dependence model, $X \subseteq \Vv$ and $\TTMS,\TTMT$ tuples over $\frM$ of length $|\Vv|$.
    Write $\TMS,\TMT$ for the induced assignments with $\TMS(\tv) = \TTMS$ and $\TMT(\tv) = \TTMT$.
    Then
    \[
        \M \models^{\FO} \vartheta_X(\TTMS,\TTMT)
        \quad\Iff\quad
        \TMS,\TMT \in \TM \text{ and } \TMS =_X \TMT.
    \]
\end{lemma}

\begin{lemma}\label{lemma:omega-sat}
    Let $\M$ and $\N$ be two $(\tau,\Vv)$ dependence models,
    so that their corresponding $(\tau \cup \{\TM\})$-structures are $\omega$-saturated.
    If $\M,\TMS \sim_{\L[\Omega]}^{\omega} \N,\TNS$, then already $\TMS \sim_{\L[\Omega]} \TNS$.
\end{lemma}
\begin{proof}
    For ease of presentation we write $\sim$ and $\sim^\omega$, omitting the subscript $\L[\Omega]$.
    Let $\M$ and $\N$ be as described above.
    By our Ehrenfeucht-Fra\"issé Theorem, it follows from $\TMS \sim^\omega \TNS$ that $\TMS \equiv_{\L[\Omega]} \TNS$. Hence it suffices to show that
    \[
        Z \coloneqq \{(\TMS,\TNS) \in \TM \times \TN \mid \M,\TMS \equiv_{\L[\Omega]} \N,\TNS\}.
    \]
    is a bisimulation between $\M$ and $\N$.
    Given $(\TMS,\TNS) \in Z$, we clearly have $\TMS \equiv^0_{\L[\Omega]} \TNS$,
    so $\TMS$ and $\TNS$ agree on atoms.
    We proceed by checking the forth condition.

    Let $\TMT \in \TM$ and $X \subseteq \comvar{\TMT}{\TMS} = \{x \in \Vv \mid \TMT =_x \TMS\}$ be finite.
    Set
    \[
        p(\ty) \coloneqq \{\vartheta_X(\ty,\TNS)\} \cup \Th_{\L[\Omega]}(\TMT)(\ty).
    \]
    We want to show that $p$ is a type of $\N$ with parameters $\TTNS$,
    i.e.~that $p$ together with the first-order theory $\Th_{\FO}(\N,\TTNS)$ is satisfiable.
    For a compactness argument we consider a finite $\Phi_0(\ty) \subseteq \Th_{\L[\Omega]}(\TMT)(\ty)$
    and define
    \[
        \phi(\tx) \coloneqq
        \exists \TTNT(\vartheta_X(\TTNT,\tx) \land \bigwedge \Phi_0(\TTNT)).
    \]
    Now there exists some finite
    $\Psi_0 \subseteq \L[\Omega]$ with $\Phi_0(\ty) = \{\phi^*(\ty) \mid \phi \in \Psi_0\}$.
    Since our translation commutes with $\land$, we obtain $(\bigwedge \Psi_0)^* = \bigwedge \Phi_0$.
    We claim that $\phi \equiv (\QE_X \bigwedge \Psi_0)^*$.
    Clearly any $(\tau \cup \{\TM\})$-structure can be interpreted as the corresponding
    structure to a $(\tau,\Vv)$ dependence model $(\frA,T'')$.
    For such a $(\frA,T'')$ and any $a \in T''$ we have
    \begin{align*}
        &(\frA,T'') \models^{\FO} \phi(\ta)\\
        \Iff\quad&\text{there exists $\tb$ over $\frA$ with $(\frA,T'') \models^{\FO} \vartheta_X(\tb,\ta)$
            and $(\frA,T'') \models^{\FO} \left(\bigwedge \Psi_0\right)^*(\tb)$}\\
        \Iff\quad&\text{there exists $b \in T''$ with $b =_X a$
            and $(\frA,T'') \models_{b} \bigwedge \Psi_0$}\\
        \Iff\quad&(\frA,T'') \models_{a} \QE_X \bigwedge \Psi_0\\
        \Iff\quad&(\frA,T'') \models^{\FO} \left(\QE_X \bigwedge \Psi_0\right)^*(\ta).
    \end{align*}
    Thus, as claimed, we have $\phi \equiv \left(\QE_X\bigwedge \Psi_0\right)^*$.
    From $\M \models_{\TMT} \bigwedge \Psi_0$ and $\TMT =_X \TMS$ we obtain
    $\M \models_{\TMS} \QE_X\bigwedge \Psi_0$.
    Since $\TMS \equiv_{\L[\Omega]} \TNS$ we get $\N \models_{\TNS} \QE_X\bigwedge \Psi_0$
    and therefore $\N \models^{\FO} \phi(\TTNS)$.
    Hence $\{\vartheta_X(\ty,\TTNS)\} \cup \Phi_0(\ty) \cup \Th_{\FO}(\N,\TTNS)$ is satisfiable.
    It follows by compactness that
    $p(\ty)$ is a type with finitely many parameters (namely $\TTNS$) over $\N$.

    By $\omega$-saturatedness we obtain some tuple $\TTNT$ in $\frN$
    with $\N \models^{\FO} p(\TTNT)$.
    From the definition of $\vartheta_X$ we see that
    $\TTNT = \TNT(\tv)$ for some $\TNT \in \TN$ with $\TNT =_X \TNS$.
    Furthermore $\N \models^{\FO} \Th_{\L[\Omega]}(\TMT)(\TTNT)$,
    so $\TMT \equiv_{\L[\Omega]} \TNT$ and hence $(\TMT,\TNT) \in Z$,
    which proves the forth condition.
    The back condition is shown analogously.
    We conclude that $Z$ is a bisimulation.
\end{proof}

\begin{theorem}[Expressive Completeness]\label{thm:exprcomp}
    For any $\phi \in \FO(\tau \cup \{\TM\})$ the following are equivalent:
    \begin{enumerate}
        \item $\phi$ is $\L[\Omega]$-bisimulation-invariant,
            i.e.~for all $\M,\TMS$ and $\N,\TNS$ of type $(\tau,\Vv)$
            \[
                \TMS \sim_{\L[\Omega]} \TNS
                \qquad\Longrightarrow\qquad
                \M \models^{\FO} \phi(\TTMS)
                \quad\text{iff}\quad
                \N \models^{\FO} \phi(\TTNS).
            \]
        \item $\phi \equiv \psi^*$ for some $\psi \in \L[\Omega](\tau,\Vv)$,
            so that for all $\M,\TMS$ of type $(\tau,\Vv)$
            \[
                \M \models^{\FO} \phi(\TTMS)
                \qquad\Iff\qquad
                \M \models_{\TMS} \psi.
            \]
    \end{enumerate}
    As this result holds a fixed $\Omega$ and arbitrary finite types $(\tau,\Vv)$,
    we write $\FO/{\sim_{\L[\Omega]}} \equiv \L[\Omega]$.
\end{theorem}
\begin{proof}
    As in the previous proof, we omit $\L[\Omega]$ from $\sim$ and $\sim^\omega$,
    and say bisimulation instead of $\L[\Omega]$-bisimulation.
    The implication ``(2) $\Rightarrow$ (1)'' is clear by our Ehrenfeucht-Fra\"issé Theorem.
    Now assume that $\phi$ is bisimulation-invariant
    and consider the set of its $\L[\Omega]$-consequences:
    \[
        C(\phi)(\tx) \coloneqq
        \{\psi^*(\tx) \mid \psi \in \L[\Omega](\tau,\Vv)\text{ and } \phi \models \psi^*\}.
    \]
    We claim that it suffices to show $C(\phi) \models \phi$.
    Indeed, using compactness this yields a finite subset $C_0 \subseteq C(\phi)$
    with $C_0 \models \phi$ and therefore $\bigwedge C_0 \equiv \phi$.
    But $C_0 = \{\psi^* \mid \psi \in \Psi\}$ for some (finite) $\Psi \subseteq \L[\Omega](\tau,\Vv)$,
    so by setting $\psi = \bigwedge \Psi$,
    we obtain the desired result that $\phi \equiv \bigwedge C_0 = \psi^*$
    for some $\psi \in \L[\Omega](\tau,\Vv)$.

    If $C(\phi)$ is unsatisfiable, $C(\phi) \models \phi$ holds vacuously.
    Hence let $\M,\TMS$ be a dependence model of type $(\tau,\Vv)$ with $\M \models^{\FO} C(\phi)(\TTMS)$.
    We need to show $\M \models^{\FO} \phi(\TTMS)$.
    If $\Th_{\L[\Omega]}(\TMS) \cup \{\phi\}$ were unsatisfiable,
    by compactness there would be some finite $\Phi_0 \subseteq \Th_{\L[\Omega]}(\TMS)$
    such that $\Phi_0 \cup \{\phi\}$ is unsatisfiable.
    This implies $\phi \models \lnot \bigwedge \Phi_0$.
    Since $\Omega$ is closed under negation, $\L[\Omega]$ is too,
    and we obtain $\lnot \bigwedge \Phi_0 \in C(\phi)$.
    This contradicts $\M \models^{\FO} C(\phi)(\TTMS)$ and $\Phi_0 \subseteq \Th_{\L[\Omega]}(\TMS)$.

    Therefore $\Th_{\L[\Omega]}(\TMS) \cup \{\phi\}$ has some model $\N,\TNS$.
    Since $\N \models^{\FO} \Th_{\L[\Omega]}(\TMS)(\TTNS)$ we obtain $\TMS \equiv_{\L[\Omega]} \TNS$.
    Now take $\omega$-saturated elementary extensions $\M \preceq (\frM_+,\TM_+)$
    and $\N \preceq (\frN_+,\TN_+)$.
    By elementary extension we have $(\frM_+,\TM_+),\TMS \equiv_{\L[\Omega]} \M,\TMS$
    and likewise for $(\frN_+,\TN_+),\TNS$. It follows that
    \[
        (\frM_+,\TM_+),\TMS
        \ \equiv_{\L[\Omega]}\ \M,\TMS
        \ \equiv_{\L[\Omega]}\ \N,\TNS
        \ \equiv_{\L[\Omega]}\ (\frN_+,\TN_+),\TNS.
    \]
    Our Ehrenfeucht-Fra\"issé Theorem yields $(\frM_+,\TM_+),\TMS \sim^{\omega} (\frN_+,\TN_+),\TNS$.
    But now we can apply \cref{lemma:omega-sat} to find
    $(\frM_+,\TM_+),\TMS \sim (\frN_+,\TN_+),\TNS$.
    Since $\phi \in \FO$, we infer $(\frN_+,\TN_+) \models^{\FO} \phi(\TTNS)$
    from $\N \models^{\FO}\phi(\TTNS)$ by elementary extension.
    Moreover, $\phi$ is bisimulation-invariant,
    so we obtain $(\frM_+,\TN_+) \models^{\FO}\phi(\TTMS)$.
    Again by elementary extension, we arrive at $\M \models^{\FO}\phi(\TTMS)$.

    With this we showed that whenever $\M \models^{\FO} \phi(\TTMS)$
    we also have $\M \models^{\FO} C(\phi)(\TTMS)$
    for an arbitrary $\M,\TMS$ of type $(\tau,\Vv)$. Hence $C(\phi) \models \phi$.
    We discussed above how this concludes the proof of the theorem.
\end{proof}

\section{Local Dependence and Guarded Fragments of \texorpdfstring{$\FO$}{FO}}
\label{sec:comp:gf}

We have seen that some of the considered local logics can be embedded into $\GF$,
the guarded fragment of first-order logic. However, it was unclear whether this
is also the case for $\LFD$, which features the local dependence atom.
We now discuss the relationship of local dependence with
$\GF$ and other guarded fragments of $\FO$.
Arguably one of the most natural generalizations of $\GF$ within $\FO$
is the clique-guarded fragment $\CGF$, introduced in \cite{Graedel99a}.

Recall that the Gaifman graph of a relational $\tau$-structure $\frA$ 
has as its universe the universe of $\frA$, and an edge between two distinct elements
$a\neq b$ if these coexist in some atomic fact of $\frA$,
i.e.~they occur together in some $\tc \in R^{\frA}$ for some $R \in \tau$.
Obviously, guarded tuples in a relational structure $\frA$ induce
a clique in the Gaifman graph of $\frA$.
Moreover, for each finite relational $\tau$ and $k \in \bbN$, there is a positive,
existential first-order formula $clique(x_1,\dots,x_k)$ which is satisfied at a tuple $\ta$
of some $\tau$-structure $\frA$ if and only if $\ta$ induces a clique in the Gaifman graph of $\frA$.
$\CGF$ is then defined in an analogous way to $\GF$,
but always uses $clique$ of the right arity as a guard, in the sense of
\[
    \forall \ty(clique(\tx\ty) \impl \phi(\tx\ty))
    \qquad\text{and}\qquad
    \exists \ty(clique(\tx\ty) \land \phi(\tx\ty)).
\]
Obviously $\ML \subsetneq \GF \subsetneq \CGF \subsetneq \FO$.
Guarded bisimulations between structures $\frA,\ta$ and $\frB,\tb$
are defined as sets $I$ of partial isomorphisms between $\frA$ and $\frB$
such that $(\ta \mapsto \tb) \in I$ and $I$ is closed under suitable back and forth conditions.
This notion naturally extends to clique-guarded bisimulation for $\CGF$, as described in \cite{Graedel99a}.
For a survey of various notions of bisimulation
and their uses for understanding expressive power,
model-theoretic and algorithmic properties of modal and guarded logics,
we refer the reader to \cite{GraedelOtt14}.

\paragraph*{Other first-order translations}
Apart from the standard translation discussed until now,
one can also consider other first-order translation of $\L[\Omega]$.
In \cite{BaltagBen20}, Baltag and van Benthem emphasized the modal perspective
of $\LFD \equiv \L[D,\anon]$, and presented a modal semantics for $\LFD$
over so-called standard relational models.
The semantics of local dependence and local independence
only require knowledge about structure of the team with respect to the ''agreement''-relations $=_X$
for $X \subseteq \Vv$.
In this sense, one introduces abstract equivalence relations $\sim_X$ which allow
us to abstract away the actual values of assignments while preserving the intended semantics.
A dependence model induces such a standard relational model in a straightforward way:
the new universe is the team, the relations become monadic and we add the equivalences $\sim_x$
for $x \in \Vv$. Formulae in the base logic $\L$ are then translated as
\[
    \tr^\TMS(R\tx) \coloneqq R_{\tx} \TMS
    \qquad\text{and}\qquad
    \tr^\TMS(\QD_X\phi) \coloneqq \forall \TMT\left(\bigwedge_{x \in X} \TMT \sim_x \TMS \impl \tr^\TMT(\phi)\right),
\]
where $\tr$ commutes with boolean connectives. For the local atoms, we set
\[
    \tr^\TMS(D_xy) \coloneqq \forall \TMT\left(\TMT \sim_x \TMS \impl \TMT \sim_y \TMS\right)
    \qquad\text{and}\qquad
    \tr^\TMS(\Ind_x y) \coloneqq \forall \TMT \exists \TMU(\TMS \sim_x \TMU \land \TMU \sim_y \TMT).
\]
If we restrict the considered class of structures to
those that are induced by dependence models under this correspondence
(so in particular they have to interpret the $\sim_x$ as equivalence relations)
then this translation shares many of the nice characteristics of the standard translation.
Indeed, using the modal translation in the definition of $\Th_{\L[\Omega]}$
and setting $\vartheta_X(\TMS,\TMT) = \bigwedge_{x \in X} \TMS \sim_x \TMT$,
the proof of the characterisation theorem in the last section
also works in this modal context, with minimal adaptions.

\paragraph*{Expressive Incomparability}
In the following we want to prove that under these translations,
local dependence is inherently incompatible with the clique-guarded fragment of first-order logic.

Let $\tau$ be a relational vocabulary and $\frA,\frB$ two $\tau$-structures
with disjoint domains $A,B$.
Their disjoint union is the $\tau$-structure $\frC = \frA \uplus \frB$ with domain 
$C = A \uplus B$ so that $\frC \restr A$ and $\frC \restr B$ are isomorphic
to $\frA$ and $\frB$ respectively, and for all tuples $\tc$ over $C$
which contain elements from both $A$ and $B$ we have $\frC \models \lnot R\tc$
for every $R \in \tau$.

\begin{proposition}\label{fact:bisim-disj}
    The relevant bisimulations for $\GF$ and $\CGF$ are compatible with
    disjoint unions of bisimilar models.
    Specifically, let $I$ be a clique-guarded bisimulation between
    $\frA$ and $\frB$, and $I'$ be one between $\frA$ and $\frC$.
    Then $I \cup I'$ is a clique-guarded bisimulation between $\frA$ and $\frB \uplus \frC$.
    In particular, from $\frA,\ta \sim_{\CGF} \frB,\tb$ and
    $\frA,\ta \sim_{\CGF} \frC,\tc$ we infer
    \[
        \frA,\ta \sim_{\CGF} (\frB\uplus\frC),\tb
        \qquad\text{and}\qquad
        \frA,\ta \sim_{\CGF} (\frB\uplus\frC),\tc.
    \]
    In this setting, if $\phi(\tx) \in \CGF$, then $\frA \models \phi(\ta)$ holds
    if and only if $(\frB \uplus \frC) \models \phi(\tb) \land \phi(\tc)$.
\end{proposition}

This highlights a small but important
difference between these guarded fragments and our logics of local dependence.
Namely, an analogue for invariance under disjoint union
cannot hold for dependence models in presence of the local dependence atom.
Indeed, we can define constancy of a variable $x$ via $D_\emptyset x$,
which is clearly not invariant under disjoint unions.

\begin{proposition}\label{prop:notgf}
    The standard and modal translations do not embed $\L[D]$ into $\CGF$.
    More specifically, there does not exist a $\CGF$-sentence $\psi$
    that is equivalent to either translation of the constancy atom $D_\emptyset x$.
\end{proposition}
\begin{proof}
    In the context of the standard translation,
    it is clear that the disjoint union of the first-order structures
    corresponding to two dependence models is itself the
    first-order structure corresponding to the disjoint union of the dependence models.
    Unlike $D_\emptyset x$, $\CGF$ is invariant under disjoint unions,
    so there cannot exist a $\psi \in \CGF$ with $(D_\emptyset x)^* \equiv \psi$.

    The argument for the modal translation is similar, and relies on the fact
    that the class of first-order structures we consider is well-behaved
    with respect to disjoint union in the above sense.
\end{proof}

There are certainly also notions expressible in $\GF$ but not in $\L[\Omega]$.
For one, in the setting of the standard translation it becomes obvious that
$\L[\Omega]$ cannot make statements about assignments and values outside of the team;
we may have $\M \models^{\FO} \exists x Rx$ and simultaneously $\M \models \QGA \lnot Rx$.
Even in the modal setting we can easily find formulae in $\GF$ that
are not equivalent to any $\LFD \equiv \L[D,\anon]$ formula.

\begin{proposition}\label{prop:notgf2}
    Let $(\tau,\Vv)$ be a finite type with $x,y,z \in \Vv$ and let
    \[
        \phi(\TMS) \coloneqq \forall \TMT(\TMT \sim_x \TMS\ \impl\ (\TMT \sim_y \TMS \lor \TMT \sim_z \TMS)) \in \GF^2.
    \]
    Then there exists no $\psi \in \L[D,\anon]$ whose modal translation is equivalent
    to $\phi$ over the considered class of first-order structures.
\end{proposition}
\begin{proof}
    We give two $(\emptyset,\{x,y,z\})$ dependence models
    that are $\L[D,\anon]$-bisimilar, but where their corresponding
    (modal) first-order structures disagree on $\phi$.
    These are uniquely determined by the structure of their teams.
    The first team has $=_y$-classes $\{t_1,s\},\{t_2\}$ and $=_z$-classes $\{t_1\},\{s,t_2\}$.
    The second team has $=_y$-classes $\{t'_1,s',s''\},\{t'_2,t''_2\}$
    and $=_z$-classes $\{t'_1\},\{s', t'_2\},\{s'', t''_2\}$.
    In both teams all assignments agree on $x$.
    The bisimulation relates $s$ to $s'$ and $s''$, $t_1$ to $t'_1$, and $t_2$ to $t'_2$ and $t''_2$.
    Then $s \sim_{\L[D,\anon]} s'$ but $\phi$ holds at $s$, while it does not hold at $s'$.
    The example is easily extended to larger types containing the variables $x,y,z$.
\end{proof}

\begin{corollary}
    In the context of the standard and modal translations,
    $\LFD \equiv \L[D,\anon]$ is expressively incomparable
    to the guarded fragment $\GF$ and even the clique-guarded fragment $\CGF$.
\end{corollary}

\section{Model checking}\label{sec:mc}

In this last section we study the complexity of the model checking problem,
abbreviated $\MC(\L[\Omega])$, for local logics $\L[\Omega]$:
Given a formula $\psi\in\L[\Omega]$ and a fitting finite pointed dependence model $\M,\TMS$,
we ask whether $\M,\TMS \models \psi$.
We consider the \emph{combined complexity}, measured with respect to the size of all inputs.
It turns that this complexity is largely influenced
by how we encode the team $\TM$ of $\M$.

For comparison, we recall that the model checking for
first-order logic (\FO) is $\PSPACE$-complete in general, 
but $\PTIME$-complete for many interesting fragments
of $\FO$ \cite[Chapter 3.1]{Graedel+07}, 
including the modal fragment $\ML$, the bounded
variable fragments $\FO^k$ (for any $k\geq 2$), and also
the guarded fragment $\GF$  \cite{BerwangerGra01}.

For the rest of this section $(\tau,\Vv)$ denotes a finite type,
$\Omega\subseteq \{D,\anon,=,\neq,\in,\notin,\Ind,\lnot \Ind\}$,
is a collection of local atoms,
and we study the model checking problem of  
$\L[\Omega](\tau,\Vv)$.

Before we discuss these technicalities of team encodings,
we consider the special case of full models.
We call dependence models full if their team consists of the whole
assignment space, so $\M$ is full if $\TM = M^{\Vv}$.
Over the class of all full dependence models,
even the base logic $\L = \L[\emptyset]$ is as expressive as relational first-order logic without equality;
the quantifiers $\exists x$ and $\forall x$ then have the same semantics
as $\QE_{\Vv_x}$ and $\QD_{\Vv_x}$ respectively, where $\Vv_x \coloneqq \Vv \setminus \{x\}$.
For example, if $\M$ is full and $\Vv = \{x,y\}$, then
\[
    \M \models^{\FO} \forall x \exists y Rxy
    \quad\Iff\quad
    \M \models \QD_{y}\QE_x Rxy
    \quad\Iff\quad
    \M \models \QGA\QE_x Rxy.
\]
This was already noted for $\LFD \equiv \L[D,\anon]$ by Baltag and van Benthem in \cite{BaltagBen20}.

It is therefore not surprising that when restricting attention to full dependence models,
we obtain the same lower bounds on the complexity of $\MC(\L[\Omega])$ as for first-order logic.
Denote by $\MCF(\L[\Omega])$ the restriction of model checking for $\L[\Omega]$
to instances where the team is always the full team,
and write $\MCFK(\L[\Omega])$ for its $k$-variable restriction.
Thus the inputs are of the form $(\psi,\frM,\tv,\TTMS)$ where $\tv$ is an enumeration of $\Vv$
and $\TTMS = \TMS(\tv)$ for $\TMS \in M^{\Vv}$ is an encoding of
the current assignment at which $\psi$ should be evaluated.
By the same techniques as for first-order logic, one
immediately obtains the following hardness results.

\begin{proposition}\label{prop:mcf-pspace}\label{prop:mcfk-ptime}
    $\MCF(\L)$ is $\PSPACE$-hard and $\MCFK(\L)$ is $\PTIME$-hard for $k \geq 2$.
    Since $\L$ is a sublogic of $\L[\Omega]$, these results also hold
    for $\L[\Omega]$ with arbitrary $\Omega$.
\end{proposition}

\subparagraph{Solving the model checking problem for $\L[\Omega]$.}
We consider a general method for solving the model checking problem for $\L[\Omega]$
that abstracts away the specifics on how the team is encoded.
We assume familiarity with alternating complexity classes
as given in \cite[Chapter 3]{BalcazarDiaGab90} or \cite[Chapters 16.2 \& 19.1]{Papadimitriou94}.
The model checking problem for first-order logic model can be solved 
in a standard way (cf.~\cite[Chapter 3.1]{Graedel+07}) by an alternating algorithm
which, to determine whether $\frA\models\psi(\ta)$, requires
\begin{itemize}
    \item alternating space $\cO(\log |\psi| + r\log|\frA|)$, where $r$ is the maximal
        number of free variables in any subformula of $\psi$, and
    \item alternating time $\cO(|\psi|\log|\frA|)$.
\end{itemize}
    
\samepage{
Together with the well-known facts that $\ALOGSPACE = \PTIME$ and $\APTIME = \PSPACE$
one then obtains that $\MC(\FO^k) \in \PTIME,\ k \in \bbN$ and $\MC(\FO) \in \PSPACE$.
It is straightforward to adapt this alternating algorithm to our setting.

\vspace{1em}
\begin{algorithm}[H]
    \small
    \DontPrintSemicolon
    \SetInd{2em}{1.5em}
    \SetAlgoCaptionSeparator{:}
    \caption{Alternating model checking for $\L[\Omega]$}\label{algo:mc}
    \textbf{ModelCheck($\psi,\M,\TMS$)}\;
    \KwInput{a formula $\psi \in \L[\Omega](\tau,\Vv)$ in negation normal form where
        $(\tau,\Vv)$ is finite, and\newline
        a finite pointed $(\tau,\Vv)$ dependence model $\M,\TMS$.}
    \vspace{1em}
    \If{\text{\normalfont $\psi$ is a literal}}
    {
        \lIf{$\M,\TMS \models \psi$}{\accept \textbf{ else} \reject}
    }
    \vspace{.5em}
    \If{$\psi = \eta \land \vartheta$}
    {
        \textbf{universally choose} $\phi \in \{\eta,\vartheta\}$\;
        \textbf{ModelCheck($\phi,\M,\TMS$)}
    }
    \vspace{.5em}
    \If{$\psi = \eta \lor \vartheta$}
    {
        \textbf{existentially guess} $\phi \in \{\eta,\vartheta\}$\;
        \textbf{ModelCheck($\phi,\M,\TMS$)}
    }
    \vspace{.5em}
    \If{$\psi = \QD_X\phi$}
    {
        \textbf{universally choose} $\TMT \in \TM$ with $\TMT =_X \TMS$\;
        \textbf{ModelCheck($\phi,\M,\TMT$)}
    }
    \vspace{.5em}
    \If{$\psi = \QE_X\phi$}
    {
        \textbf{existentially guess} $\TMT \in \TM$ with $\TMT =_X \TMS$\;
        \textbf{ModelCheck($\phi,\M,\TMT$)}
    }
\end{algorithm}
}
It implements the usual model checking game
between an existential and universal player,
played on positions $(\phi,\TMT)$ where $\phi$ is some subformula of $\psi$
and $\TMT \in \TM$ an assignment.
The algorithm \accept s if and only if the existential player has
a winning strategy for this game, which is the case if and only if
$\M,\TMS \models \psi$.

\subparagraph{How to encode the team.}
Now we come back to the discussion on how to encode the team.
A first idea might be to simply list all
assignments of the team in the input.
Call this variant $\MCL(\L[\Omega])$.
In many cases, this encoding is rather inefficient;
given $\frM,\psi$, a tuple $\tv$ enumerating $\Vv$ and a current assignment $\TMS \in \TM$,
encoding $\TM$ as a list of all its assignments may cause an exponentially longer input,
as in the case of full models where $\TM = M^{\Vv}$.
As a consequence, one obtains a deceptively low complexity for $\MCL(\L[\Omega])$.

\begin{proposition}\label{prop:algo-comp1}
    The alternating algorithm can be implemented to decide instances
    $(\psi,\frM,\tv,\TTMS,\TM)$ of $\MCL(\L[\Omega])$
    with alternating workspace $\cO(\log|\psi| + \log |\TM|)$.
    In particular, we obtain that $\MCL(\L[\Omega]) \in \ALOGSPACE = \PTIME$.
\end{proposition}
\begin{proof}
    We implement picking (existentially guessing or universally choosing)
    assignments $\TMT \in \TM$ by picking a pointer to some tuple in the list representing $\TM$.
    Such a pointer only requires $\log |\TM|$ space,
    which allows us to decide whether $\M,\TMS \models \beta$ with alternating workspace
    $\cO(\log |\TM|)$, for any $\beta \in \Omega$. This is clear for $\beta \in \{=,\neq\}$.
    For $\beta = D_Xy$, we universally choose $\TMT \in \TM$ with $\TMT =_X \TMS$ and
    {\accept} if and only if $\TMT =_y \TMS$.
    If $\beta = (\tx \in \ty)$, then we existentially guess $\TMT \in \TM$
    and {\accept} if and only if $\TMS(\tx) = \TMT(\ty)$.
    Given $\beta = \Ind_{\tx}\ty$, we universally choose $\TMU \in \TM$ and then
    existentially guess $\TMT \in \TM$, accepting
    only in the case that $\TMS(\tx) = \TMT(\tx)$ and $\TMU(\ty) = \TMT(\ty)$.
    The atoms $\anon,\notin,\lnot\Ind$ are handled dually to their counterparts.

    A pointer of length $\log |\psi|$ suffices to specify the current subformula $\phi$ of $\psi$.
    Since $\psi,\frM$ and $\TM$ are never modified,
    the algorithm only needs to keep track of at most 3 assignments and 3 subformulae at any time.
    Together with the above analysis, this proves the claim.
\end{proof}

Comparing the above to our result that $\MCF(\L[\Omega])$ is $\PSPACE$-hard,
we seem to have a contradiction to the common belief that $\PTIME \neq \PSPACE$.
This is however not the case, as the two problems differ on the length of their inputs.
Indeed, there cannot exist a polynomial-time reduction from $\MCF(\L[\Omega])$
to $\MCL(\L[\Omega])$, because the size of the full team $M^{\Vv}$ is exponential in the size
$|\psi| + |\frM| + |\Vv|$ of the input of $\MCF(\L[\Omega])$.

Nevertheless, this disparity shows that the approach of
encoding $\TM$ as a list may yield unsatisfactory results.
A different approach is to encode the team as a first-order formula $\phi_\TM$
over the vocabulary $\tau \cup M$ where all elements of $\frM$ are added as constants
interpreted by themselves in the corresponding expansion $\frM_M$ of $\frM$.
We want that for all $\TTMT \in M^{|\Vv|}$:
\[
    \frM_M \models \phi_\TM(\TTMT)
    \qquad\Iff\qquad
    \text{there is some $\TMT \in \TM$ with $\TMT(\tv) = \TTMT$.}
\]

\begin{Definition}\label{def:mcfo}
    We encode the team by a first-order formula as described above.
    More specifically, define the problem $\MCFO(\L[\Omega])$ as follows.
    \begin{enumerate}
        \item The inputs are tuples $(\psi,\frM,\tv,\TTMS,\phi_\TM)$ where
            \begin{enumerate}
                \item $(\psi,\frM,\tv,\TTMS)$ have the same interpretation as before,
                    with $\psi \in \L[\Omega](\tau,\Vv)$, $\frM$ a $\tau$-structure with universe $M$,
                    $\tv$ an enumeration of $\Vv$ and $\TTMS \in M^{|\Vv|}$ encoding
                    a current assignment.
                \item $\phi_\TM \in \FO(\tau \cup M)$ encodes the team $\TM \subseteq M^{\Vv}$
                    as described above, with $\frM_M \models \phi_\TM(\TTMS)$.
            \end{enumerate}
        \item The task is to decide whether $\M,\TMS \models \psi$.
        \item The complexity is measured in the input size,
            so essentially with respect to $|\psi| + |\frM| + |\Vv| + |\phi_\TM|$.
    \end{enumerate}
\end{Definition}

\begin{proposition}\label{prop:algo-comp2}
    We can implement the alternating algorithm so that a given instance $(\psi,\frM,\tv,\TTMS,\phi_\TM)$
    for $\MCFO(\L[\Omega])$ is solved requiring only 
    \begin{enumerate}
        \item alternating space $\cO(\log |\psi| + \log |\phi_\TM| + (|\Vv| + r)\log |\frM|)$, where
            $r$ is the maximal number of free variables in any subformula of $\phi_\TM$, and
        \item alternating time $\cO(|\psi|\cdot(|\psi|+|\frM|+(|\Vv|+|\phi_\TM|)\log|\frM|))$.
    \end{enumerate}
    In particular, we obtain that $\MCFO(\L[\Omega]) \in \APTIME = \PSPACE$.
\end{proposition}
\begin{proof}
    Assignments $\TMT$ are encoded by their values $\TTMT = \TMT(\tv)$,
    thus taking $\cO(A) \coloneqq \cO(|\Vv|\log|\frM|)$ space.
    We implement picking assignments $\TMT \in \TM$ by picking a tuple $\TTMT \in M^{|\Vv|}$
    and then performing first-order model checking on $(\phi_\TM,\frM_M,\TTMT)$.
    We know that $|\frM_M|$ is polynomial in $|\frM|$ and hence $\log |\frM_M| \in \cO(\log|\frM|)$.
    It follows from the complexity of first-order model checking listed in the last paragraph
    that picking an assignment in this way requires
    alternating space $\cO(P_s) \coloneqq  \cO(A + \log |\phi_\TM| + r\log|\frM|)$ and
    alternating time  $\cO(P_t) \coloneqq \cO(A + |\phi_\TM|\log|\frM|)$.

    The space-analysis is parallel to the proof of \cref{prop:algo-comp1}.
    We store a subformula $\phi$ of $\psi$ as a pointer of length $\log |\psi|$.
    Then $\psi,\frM$ and $\phi_\TM$ are never modified,
    and at any time we need at most 3 assignments and 3 subformulae in the workspace.
    This yields the claimed alternating space-complexity $\cO(P_s + \log |\psi|)$,
    since we can always reuse the space required for the first-order model checking.

    For the time analysis, note that checking something such as $\TMT(\tx) = \TMS(\ty)$
    for two assignments $\TMS,\TMT$ in the workspace is
    possible within alternating time $\cO(A)$, the size of the assignments.
    \begin{enumerate}
        \item If $\phi$ is a relational literal
            we can evaluate whether $\M,\TMS \models \phi$
            in alternating time $\cO(|\psi| + |\frM| + A)$.
        \item For other local atoms $\beta \in \Omega$
            or dependence quantifiers,
            the algorithm picks at most two new assignments
            and then does a constant number of checks of the form $\TMS(\tx) = \TMT(\ty)$.
            In the case of dependence quantifiers,
            it also updates the current subformula from $\QD_X\phi$ or $\QE_X\phi$ to $\phi$.
            This can be accomplished in alternating time $\cO(P_t + |\psi|)$.
        \item Choosing some subformula at conjunctions and disjunctions takes only
            $\cO(|\psi|)$ time, since we just have to move our pointer within $\psi$.
    \end{enumerate}
    Clearly we have at most $|\psi|$ recursive calls,
    which yields the claimed alternating time-complexity $\cO(|\psi|(P_t + |\psi| + |\frM|))$.
\end{proof}

We now want to show an analogue of $\MC(\FO^k) \in \ALOGSPACE = \PTIME$,
so we define $\MCFOK(\L[\Omega])$ as the restriction
of $\MCFO(\L[\Omega])$ to instances where $|\Vv| \leq k$.
We now want the alternating space-complexity
given in \cref{prop:algo-comp2} to be logarithmic in the input.
Problematic is the occurrence of $r$, which describes the maximum
number of free variables in any subformula of $\phi_\TM$,
and originates from the space-complexity of the first-order model checking
we perform when picking new assignments.
Currently, we allow arbitrary $\phi_\TM \in \FO(\tau \cup M)$ 
to represent the team $\TM$ in the input.
To obtain our wanted analogue,
we need to bound $r$ by some constant for all instances of $\MCFOK(\L[\Omega])$.

Every team $\TM \subseteq M^{\Vv}$ can be encoded by an $\FO(\tau \cup M)$-formula
$\phi_\TM$ that uses only the variables in $\Vv$.
Indeed, we can just set $\phi_\TM(\tx) = \bigvee_{\TMT \in \TM} \tx = \TMT(\tv)$.
This shows that it is very lenient to assume
that there exists some global bound for $r$ in all instances of $\MCFOK(\L[\Omega])$.

\begin{Definition}
    For $B \geq k \in \bbN$ define $\MCFOBK(\L[\Omega])$ as the restriction
    of $\MCFO(\L[\Omega])$ to instances $(\psi,\frM,\tv,\TTMS,\phi_\TM)$ where:
    \begin{enumerate}
        \item $|\Vv| \leq k$, and
        \item every subformula of $\phi_\TM$ has at most $B$ free variables.
    \end{enumerate}
\end{Definition}

\begin{corollary}\label{corollary:mcfobk}
    $\MCFOBK(\L[\Omega]) \in \ALOGSPACE = \PTIME$ for all $B \geq k \in \bbN$.
\end{corollary}
\begin{proof}
    From \cref{prop:algo-comp2} and the above definition of $\MCFOBK(\L[\Omega])$
    we see that we can solve instances $(\psi,\frM,\tv,\TTMS,\phi_\TM)$ of $\MCFOBK(\L[\Omega])$
    with the algorithm from \cref{prop:algo-comp2}
    requiring only alternating space $\cO(\log |\psi| + \log|\frM| + \log|\phi_\TM|)$.
\end{proof}

\begin{proposition}~
    \begin{enumerate}
        \item $\MCFO(\L[\Omega])$ is $\PSPACE$-complete.
        \item $\MCFOBK(\L[\Omega])$ is $\PTIME$-complete for all $B \geq k \geq 2$.
    \end{enumerate}
\end{proposition}
\begin{proof}
    Since the full team is specified by $\phi_{M^{\Vv}} = \text{True}$,
    we obtain the following logspace-computable reduction 
    \[
        \MCF(\L[\Omega]) \to \MCFO(\L[\Omega]),\quad (\psi,\frM,\tv,\TTMS) \mapsto (\psi,\frM,\tv,\TTMS,\text{True}),
    \]
    which shows that $\MCF(\L[\Omega]) \leq_{\log} \MCFO(\L[\Omega])$.
    Via the same reduction we can show that
    $\MCFK(\L[\Omega]) \leq_{\log} \MCFOBK(\L[\Omega])$ for all $B \geq k \in \bbN$.
    Hence the hardness-results follow from \cref{prop:mcf-pspace,prop:mcfk-ptime}.
    The rest was already discussed in \cref{prop:algo-comp2,corollary:mcfobk}.
\end{proof}

This shows that if the local atoms are efficiently checkable,
the complexity of model checking $\L[\Omega]$ depends mostly on how one encodes the team.
Encoding the team as a list, as one would do for ordinary relations,
we obtain $\PTIME$-completeness for both the finite-variable and the unconstrained variant.
Encoding the team as a first-order formula is more efficient,
and yields essentially the same complexity as that of first-order model checking;
$\PSPACE$-complete in general, but $\PTIME$-complete in restriction to $k \geq 2$ variables.



\end{document}